%% file: Ergebnis3.tex
\documentclass[letterpaper,11pt]{article}

\usepackage{lmodern}  

\usepackage[english]{babel}
\usepackage[T1]{fontenc}
\usepackage[utf8]{inputenc}

\usepackage{amsmath}
\usepackage{amsthm}
\usepackage{amsfonts}
\usepackage{amssymb}
\usepackage{thmtools}
\usepackage[retainorgcmds]{IEEEtrantools}
\usepackage{units}
\usepackage[babel]{microtype}

\usepackage[vlined,ruled]{algorithm2e}
\usepackage{graphicx}
\usepackage{subfig}

\usepackage[hmargin=1in,vmargin=1.3in]{geometry}

\newcommand{\natur}{\mathbb{N}}
\newcommand{\menge}[1]{\left\{#1\right\}}
\newcommand{\menges}[1]{\{#1\}}

\newcommand{\betrs}[1]{|#1|}

\newcommand{\klammer}[1]{\left(#1\right)}
\newcommand{\klammers}[1]{(#1)}

\newcommand{\intervLr}[1]{\left[ #1 \right)}
\newcommand{\eps}{\varepsilon}
\newcommand{\beps}{\bar{\eps}}
\newcommand{\freps}{\frac{1}{\eps}}
\newcommand{\frepse}[1]{\frac{1}{\eps^{#1}}}

\newcommand{\unterks}[1]{\lfloor #1 \rfloor}

\newcommand{\oberks}[1]{\lceil #1 \rceil}

\newcommand{\OPT}{\ensuremath{\mathrm{OPT}}}

\newcommand{\OPTvvs}[2]{\OPT\klammers{#1, #2}}

\newcommand{\OPTvs}[1]{\OPTvvs{#1}{v}}

\newcommand{\OPTvv}[2]{\OPT\klammer{#1, #2}}

\newcommand{\OPTv}[1]{\OPTvv{#1}{v}}

\newcommand{\OPTkkvvs}[3]{\OPT_{\leq #1}\klammers{#2, #3}}

\newcommand{\OPTkkvs}[2]{\OPTkkvvs{#1}{#2}{v}}

\newcommand{\OPTkvs}[1]{\OPTkkvs{k}{#1}}

\newcommand{\OPTkkvv}[3]{\OPT_{\leq #1}\klammer{#2,#3}}

\newcommand{\OPTkkv}[2]{\OPTkkvv{#1}{#2}{v}}

\newcommand{\OPTSv}[1]{\OPT_{\mathrm{St}}\klammer{#1,v}}

\newcommand{\OPTSvs}[1]{\OPT_{\mathrm{St}}\klammers{#1,v}}

\newcommand{\sumn}{\sum_{j=1}^n}

\newcommand{\Oh}[1]{O\klammer{#1}}
\newcommand{\Ohs}[1]{O(#1)}

\newcommand{\fraczk}{\frac{1}{2^{\kappa-1}}}
\newcommand{\fraclogeps}{\frac{1}{\log_2(\frac{2}{\eps}) + 1} }
\newcommand{\efraclogeps}[1]{\frac{#1}{\log_2(\frac{2}{\eps}) + 1}}

\newcommand{\zke}[1]{2^{#1}}

\newcommand{\ta}{\tilde{a}}
\newcommand{\ba}{\bar{a}}
\newcommand{\bas}{\ba'}

\newcommand{\bp}{\bar{p}}
\newcommand{\bs}{\bar{s}}
\newcommand{\tp}{\tilde{p}}
\newcommand{\ts}{\tilde{s}}

\newcommand{\ame}{a_{\mathrm{meff}}}
\newcommand{\aeff}{a_{\mathrm{eff}}}
\newcommand{\aeffc}{a_{\mathrm{eff-c}}}
\newcommand{\akg}{a^{(k)}_\gamma}
\newcommand{\aege}[2]{a^{(#1)}_{#2}}
\newcommand{\takg}{\ta^{(k)}_\gamma}
\newcommand{\taege}[2]{\ta^{(#1)}_{#2}}
\newcommand{\taeg}[1]{\taege{#1}{\gamma}}

\newcommand{\tI}{G}
\newcommand{\ttI}{H}

\newcommand{\ILred}{I_{L,\mathrm{red}}}
\newcommand{\tIk}{\tI^{(k)}}
\newcommand{\tIe}[1]{\tI^{(#1)}}

\newcommand{\ttIe}[1]{\ttI^{(#1)}}
\newcommand{\Ik}{I^{(k)}}
\newcommand{\Ie}[1]{I^{(#1)}}

\newcommand{\Lk}{L^{(k)}}
\newcommand{\Le}[1]{L^{(#1)}}
\newcommand{\Lkg}{\Lk_{\gamma}}
\newcommand{\Leg}[1]{\Le{#1}_\gamma}
\newcommand{\Lege}[2]{\Le{#1}_{#2}}
\newcommand{\tLkx}{\tilde{L}^{(\kappa-2)}_{\xi}}
\newcommand{\tLkxe}[1]{\tilde{L}^{(\kappa-2)}_{#1}}

\newcommand{\Fk}{F^{(k)}}
\newcommand{\Fke}[1]{F^{(#1)}}
\newcommand{\Dk}{D^{(k)}}
\newcommand{\Dke}[1]{D^{(#1)}}
\newcommand{\tDk}{\tilde{D}^{(k)}}

\newcommand{\backtrack}{\mathrm{Backtrack}}

\newtheorem{lemma}{Lemma}
\theoremstyle{definition} \newtheorem{remark}[lemma]{Remark}
\theoremstyle{plain} \newtheorem{theorem}[lemma]{Theorem}
\newtheorem{corollary}[lemma]{Corollary}
\newtheorem{definition}[lemma]{Definition}
\newtheorem{assumption}{Assumption}

\pdfpagewidth=\paperwidth
\pdfpageheight=\paperheight

\usepackage[backend=bibtex,style=numeric,citestyle=numeric-comp,bibstyle=numeric,sorting=nyt,firstinits=true,sortcites=true,defernumbers=true,maxnames=10]{biblatex}
\AtEveryBibitem{\clearname{editor}}
\usepackage{csquotes}
\newcommand{\Citewa}[1]{\Citeauthor{#1}~\cite{#1}}

\title{A Faster FPTAS for the Unbounded Knapsack Problem\thanks{Research supported by DFG project JA612/14-2, ``Entwicklung und Analyse von effizienten polynomiellen Approximationsschemata f\"ur Scheduling- und verwandte Optimierungsprobleme''}}
\author{Klaus Jansen \qquad Stefan E.\@ J.\@ Kraft\\
Department of Computer Science, Kiel University, 24098 Kiel, Germany\\
{\ttfamily\{kj,stkr\}@informatik.uni-kiel.de}}
\date{}

\usepackage[pdftex, 
colorlinks = false]{hyperref}

\begin{document}
\maketitle

\input{introduction}
\input{preliminaries}
\input{simplify_structure}
\input{dynamic_programming}
\input{complete_algorithm}

\printbibliography
This bibliography contains information from the DBLP database (\url{www.dblp.org}), which is made available under the ODC Attribution License.
\end{document}

%% file: introduction.tex
\begin{abstract}
The Unbounded Knapsack Problem (UKP) is a well-known variant of the famous 0-1 Knapsack Problem (0-1 KP). In contrast to 0-1 KP, an arbitrary number of copies of every item can be taken in UKP. 
Since UKP is NP-hard, fully polynomial time approximation schemes (FPTAS) are of great interest. Such algorithms find a solution arbitrarily close to the optimum $\OPT(I)$, i.e.\@ of value at least $(1-\varepsilon) \OPT(I)$ for $\varepsilon > 0$, and have a running time polynomial in the input length and $\frac{1}{\varepsilon}$.
For over thirty years, the best FPTAS was due to Lawler with a running time in $\Ohs{n + \frac{1}{\varepsilon^3}}$ and a space complexity in $\Ohs{n + \frac{1}{\varepsilon^2}}$, where $n$ is the number of knapsack items. 
We present an improved FPTAS with a running time in $\Ohs{n + \frac{1}{\varepsilon^2} \log^3 \frac{1}{\varepsilon}}$ and a space bound in $\Ohs{n + \frac{1}{\varepsilon} \log^2 \frac{1}{\varepsilon}}$. This directly improves the running time of the fastest known approximation schemes for Bin Packing and Strip Packing, which have to approximately solve UKP instances as subproblems. 
\end{abstract}
\section{Introduction}
An instance $I$ of the Knapsack Problem (KP) consists of a list of $n$ items $a_1, \ldots, a_n$, $n \in \natur$, where every item has a profit $p_j \in (0,1]$ and a size $s_j \in (0,1]$. Moreover, we have the knapsack size $c = 1$. In the 0-1 Knapsack Problem (0-1 KP), a subset $V \subset \menge{a_1, \ldots, a_n}$ has to be chosen such that the total profit of $V$ is maximized and the total size of the items in $V$ is at most $c$. Mathematically, the problem is defined by $\max\menges{\sumn p_j x_j | \sumn s_j x_j \leq c; x_j \in \menges{0,1} \ \forall j}$. 
In this paper, we focus on the unbounded variant (UKP) where an arbitrary number of copies of every item is allowed, i.e.\ we want to determine $\max\menges{\sumn p_j x_j | \sumn s_j x_j \leq c; x_j \in \natur \ \forall j}$.
\subsection{Known Results} \label{subsec:known_results}
The 0-1 Knapsack Problem and other variants of KP are well-known NP-hard problems \cite{Garey1979}. They can be optimally solved in pseudo-polynomial time by dynamic programming \cite{Bellman1957,Kellerer2004}. Furthermore, fully polynomial time approximation schemes (FPTAS) are known for different variants of KP. An FPTAS is a family of algorithms $(A_\eps)_{\eps > 0}$, where for every $\eps > 0$ the algorithm $A_\eps$ finds for a given instance $I$
a solution of profit $A_\eps(I) \geq (1-\eps) \OPT(I)$. The value $\OPT(I)$ denotes the optimal value for $I$. FPTAS have a running time polynomial in $\freps$ and the input length.

The first FPTAS for 0-1 KP was presented by Ibarra and Kim \cite{Ibarra1975} with a running time in $\Ohs{n \log n + \frepse{2} \cdot \min\menges{\frepse{2} \log (\freps),n}}$ and a space complexity in $\Ohs{n + \frepse{3}}$. 
\Citewa{Lawler1979} improved the running time to $\Ohs{\frepse{4} + \log(\freps)n}$. In 1981, \Citewa{Magazine1981} presented a method to decrease the space complexity of the dynamic program so that their FPTAS runs in time $\Ohs{n^2 \log(n) \frepse{2}}$ and needs space in $\Ohs{\frac{n}{\eps}}$. (The paper focuses on the improved space complexity without a partitioning and reduction of the items as done e.g.\ by Lawler. Without it, Lawler's basic algorithm has in fact a time and space complexity in $\Ohs{\frac{n^2}{\eps}}$.) The currently fastest known algorithm is due to Kellerer and Pferschy \cites{Kellerer1999}{Kellerer2004a}[pp.~166--183]{Kellerer2004} with a space bound in $\Ohs{n + \frepse{2}}$ and a running time in $\Ohs{n \min\menges{\log n, \log \freps} + \frepse{2}\log(\freps)\cdot \min\menges{n, \freps \log(\freps)}}$. Assuming that $n \in \Omega(\freps \log \freps)$, this is in $\Ohs{n \log(\freps) + \frepse{3} \log^2(\freps)}$.


For UKP, Ibarra and Kim \cite{Ibarra1975} presented the first FPTAS by extending their 0-1 KP algorithm. Their UKP algorithm has a running time in $\Ohs{n + \frepse{4} \log \freps}$ and a space complexity in $\Ohs{n + \frepse{3}}$.  Kellerer et al.\@ \cite[pp.~232--234]{Kellerer2004} have moreover described an FPTAS with a running time in $\Ohs{n \log (n) + \frepse{2} (n + \log \freps)}$ and a space bound in $\Ohs{n + \frepse{2}}$. In 1979, Lawler \cite{Lawler1979} presented his FPTAS with a running time in $\Ohs{n + \frepse{3}}$ and a space complexity in $\Ohs{n + \frepse{2}}$. For $n \in \Omega(\freps)$, this is still the best known FPTAS.

The study of KP is not only interesting in itself, it is moreover motivated by column generation for optimization problems like the famous Bin Packing Problem and Strip Packing Problem. In the former problem, a set $J$ of $n$ items of size in $(0,1]$ has to be packed in as few unit-sized bins as possible. In the latter problem, a set $J$ of $n$ rectangles of width $(0,1]$ and height $(0,1]$ has to be packed in a strip of unit width such that the height of the packing is minimized. Many algorithms for optimization problems like Bin Packing have to solve linear programs (LPs), but enumerating all columns of the linear programs would take too much time. One way to avoid this is the consideration of the dual of the LP and to (approximately or exactly) solve a separation problem, e.g.\ KP, to find violated inequalities of the dual. These inequalities correspond to columns in the primal LP: the columns needed for solving the LP are therefore generated and added dynamically. Examples can be found in \mbox{\cite{Gilmore1961,Karmarkar1982}}.

Since Bin Packing and Strip Packing are NP-complete \cite{Garey1979}, several approximation algorithms have been found for both problems. However, no efficient (i.e.\ polynomal-time) algorithm $A$ for BP or SP can achieve $A(J) \leq c \cdot \OPT(J)$ for $c < \frac{3}{2}$ and all problem instances $J$ unless P = NP \cite{Garey1979}: we have $c \geq \frac{3}{2}$ for the absolute approximation ratio $c$. The bound $\frac{3}{2}$ is due to the fact that a polynomial algorithm could otherwise distinguish between the optimum of 2 or 3 for BP instances and therefore solve the NP-complete Partition Problem in polynomial time \cite{Garey1979}. Since only such small instances prevent an absolute ratio better than $\frac{3}{2}$, larger instances may allow for a better approximation ratio.

So-called asymptotic fully polynomial-time approximation schemes (AFPTAS) $(A_\eps)_{\eps > 0}$ are therefore especially interesting. They find for every $\eps > 0$ and instance $J$ a solution of value at most $(1 + \eps)\OPT(J) + f(\freps)$, and have a running time polynomial in the input length and $\freps$. Roughly speaking, AFPTAS achieve an approximation ratio of $c = (1+\eps)$ for large problem instances.

For Bin Packing, the first AFPTAS was presented by Karmarkar and Karp \cite{Karmarkar1982} with $f(\freps) = \Ohs{\frepse{2}}$. In 1991, Plotkin et al.\ \cite{Plotkin1995} described an improved algorithm with a smaller additive term $f(\freps) = \Ohs{\frac{1}{\eps} \log(\frac{1}{\eps})}$ and running time in $\Ohs{\frepse{6} \log^6 (\freps) + \log(\frac{1}{\eps}) n}$. 
The AFPTAS by Shachnai and Yehezkely \cite{Shachnai2007} has the same additive term and a running time in $\Ohs{\frepse{6} \log^3(\freps) + \log(\frac{1}{\eps}) n}$ for general instances. 
Currently, the AFPTAS in \cite{Jansen2012} has the smallest additive term $f(\freps) = \Ohs{\log^2 \freps}$ and the fastest running time in $\Ohs{ \frepse{6} \log \freps + \log\klammers{\frac{1}{\eps}} n }$.

The first AFPTAS for Strip Packing was presented by Kenyon and R{\'e}mila \cite{Kenyon2000} with $f(\freps) = \Ohs{\frepse{2}}$. Bougeret et al.\ \cite{Bougeret2011} and Sviridenko \cite{Sviridenko2012} independently improved the additive term to $f(\freps) = \Ohs{\freps \log \freps}$. The algorithm in \cite{Bougeret2011} needs time in $\Ohs{\frepse{6} \log (\freps) + n \log n}$, which is the currently fastest known AFPTAS. 

Both algorithms in \cite{Bougeret2011,Jansen2012} solve UKP instances for column generation. A faster FPTAS for UKP therefore directly yields faster AFPTAS for Bin Packing and Strip Packing.
\subsection{Our Result} \label{subsec:our_result}
We have derived an improved FPTAS for UKP that is faster and needs less space than Lawler's algorithm.
\begin{theorem}
There is an FPTAS for UKP with a running time in $\Ohs{n + \frepse{2} \log^3(\freps)}$ and a space complexity in $\Ohs{n + \freps \log^2(\freps)}$.
\end{theorem}
Not only the improved running time, but also the improved space complexity is interesting because ``for higher values of $\freps$ the space requirement is usually considered to be a more serious bottleneck for practical applications than the running time'' \cite[p.~168]{Kellerer2004}. Nevertheless, the improved time complexity has direct practical consequences. 
Let $KP(d, \eps)$ be the running time to find a $(1-\eps)$ approximate solution to a UKP instance with $d$ items. The Bin Packing algorithm in \cite{Jansen2012} has the running time $\Ohs{ KP\klammers{d, \frac{\beps}{6}} \cdot \frepse{3} \log \freps + \log\klammers{\frac{1}{\eps}} n }$ if we assume that $KP\klammers{d, \frac{\beps}{6}} \in \Omega(\frepse{2})$ (where $\beps \in \Theta(\eps)$ and $d \in \Ohs{\freps \log \freps}$). By using the new FPTAS for UKP, we get the following result:
\begin{corollary}
There is an AFPTAS $(A_\eps)_{\eps > 0}$ for Bin Packing that finds for $\eps \in (0,\frac{1}{2}]$ a packing of $J$ in $A_\eps(J) \leq (1 + \eps) \OPT(J) + \Ohs{\log^2(\frac{1}{\eps})}$ bins. Its running time is in
	\[\Oh{  \frepse{5} \log^4 \freps + \log\klammer{\frac{1}{\eps}} n }\enspace.
\]
\end{corollary}
Similarly, the Strip Packing algorithm in \cite{Bougeret2011} (see also \cite{Jansen2006}) has a running time in $\Ohs{d (\frac{1}{\eps^2} + \ln d) \max \menges{KP(d, \frac{\beps}{6}), d \ln \ln (\frac{d}{\eps})} + n \log n}$ where again $d \in \Ohs{\freps \log \freps}$ and $\beps \in \Theta(\eps)$. The new FPTAS yields the following improved AFPTAS:
\begin{corollary}
There is an AFPTAS $(A_\eps)_{\eps > 0}$ for Strip Packing that finds a packing for $J$ of total height $A_\eps(J) \leq (1 + \eps) \OPT(J) + \Ohs{\freps \log(\frac{1}{\eps})}$. Its running time is in
	\[\Oh{  \frepse{5} \log^4 \freps + \log\klammer{\frac{1}{\eps}} n }\enspace.
\]
\end{corollary}
The result in this paper was first presented at IWOCA 2015 \cite{Jansen2015_16}. The final publication will be available at \url{link.springer.com}.

For readers acquainted with column generation or linear programs, it should be noted that the LP solved has the form $\min\menges{c^T x \ | \ A x \geq b, x \geq 0 }$. It is indeed a fractional covering problem where the columns of $A$ represent configurations: a configuration assigns item slots to one bin (for Bin Packing) or to one shelf of the strip (for Strip Packing) such that the slots fit into the bin or the strip. The primal LP is then approximately solved with a method by Grigoriadis et al.\ \cite{Grigoriadis2001} (see also \cite{Jansen2006}). The columns (i.e.\ configurations) are generated by solving so-called block problems, which are UKP instances in this case. When the LP has been solved, each item is placed in a slot that has at least the size of the item. As a feasible solution to the LP has been found, there are enough slots for all items. Because of the unboundedness, some configurations may indeed assign more item slots of a certain size to the strip or to one or several bins than there are items in the considered Strip or Bin Packing instance. This does not represent a problem because the supernumerary item slots are simply left empty in the final solution. For comparison, Plotkin et al.\@ \cite{Plotkin1995} solve the LP with a decomposition method where the block problem has additional constraints on the knapsack variables: it is a Bounded Knapsack Problem where a limited number $d_j \in \natur$ of copies for every item $a_j$ may be taken.

\subsection{Techniques}
Most algorithms for UKP \cite{Ibarra1975,Lawler1979,Kellerer2004} rely on 0-1 KP algorithms. The 0-1 KP algorithms determine a first lower bound $P_0$ for $\OPT(I)$. Based on a threshold $T$ depending on $P_0$, the items are partitioned into large(-profit) items with $p_j \geq T$ and small(-profit) items with $p_j < T$. A subset of the large items is taken, which is sufficient for an approximate solution. Its profits are then scaled and the well-known dynamic programming by profits applied to the subset. All combinations of large items (packed by the dynamic program) and small items (which are greedily added) are checked and the best one returned. For UKP, copies of the items in the reduced large item set are taken to transform the UKP instance into a 0-1 KP instance.

Our algorithm also first reduces the number of large items. However, we further preprocess the remaining large items by taking advantage of the unboundedness: large items of similar profit $[2^k T, \zke{k+1} T)$ are iteratively combined (``glued'') together to larger items. Apart from two special cases that can be easily solved, we prove for this new set $\tI$ a structure property: there are approximate solutions where at most one large item from every interval $[\zke{k} T, \zke{k+1} T)$ is used, i.e.\ only $\Ohs{\log \freps}$ items in total. As a next step, a large item $\aeffc$ that consists of several copies of the most efficient small item $\aeff$ is introduced. We prove that there are now approximate solutions to the large items $\tI \cup \menges{\aeffc}$ of cardinality $\Ohs{\log \freps}$ and that additionally use at least one item of profit at least $\frac{1}{4} P_0$. Instead of exact dynamic programming, we use approximate dynamic programming: the profits in $[\frac{1}{4} P_0,2 P_0]$ are divided into intervals of equal length. During the execution of the dynamic program, we eliminate dominated solutions and store for each interval at most one solution of smallest size. The combination of approximate dynamic programming with the structure properties yields the considerable improvement in the running time and the space complexity. The algorithm then returns the best combination of large items (packed by the dynamic program) and copies of the small item $\aeff$ (added greedily).


%% file: preliminaries.tex
\section{Preliminaries}
We introduce some useful notation. The profit of an item $a$ is denoted by $p(a)$ and its size by $s(a)$. If $a = a_j$, we also write $p(a_j) = p_j$ and $s(a_j) = s_j$. Let $V = \menges{x_a : a \ | \ a \in I, x_a \in \natur}$ be a multiset of items, i.e.\@ a subset of items in $I$ with their multiplicities. We naturally define the total profit $p(V) := \sum_{x_a > 0} p(a) x_a$ and the total size $s(V) := \sum_{x_a > 0} s(a) x_a$.

Let $v \leq c=1$ be a part of the knapsack. The corresponding optimum profit for the volume $v$ is denoted by $\OPTv{I} = \max\{\sum_{a \in I} p(a) x_a \ | \ \sum_{a \in I} s(a) x_a \leq v; \: a \in I; \: x_a \in \natur\}$. Obviously, $\OPT(I) = \OPTvv{I}{c}$ holds.

We assume throughout the paper that basic arithmetic operations as well as computing the logarithm can be performed in $\Oh{1}$.

Finally, we have a remark about the use of ``item'' and ``item copy'' when we consider a solution to a UKP instance.
{\newcommand{\tII}{\tilde{I}}
\begin{remark}
Let $I, \tII$ be two sets of knapsack items with $\tII \subseteq I$. In the 0-1 Knapsack Problem, a sentence like ``the solution to $I$ uses at most one item in $\tII$'' is obvious: if the solution uses one item in $\tII$, all other items of the solution are in $I \setminus \tII$.

Consider now UKP. When we talk about solutions, we would formally have to distinguish between an item $a' \in I$ in the instance and the item copies of $a'$ that a solution $V = \menges{x_a : a \ | \ a \in I, x_a \in \natur}$ uses. In this paper, we however use the expressions ``item'' and ``item copy'' interchangeably when talking about solutions. As an example, let us consider the sentence ``the solution to $I$ uses at most one item in $\tII$.'' It means that the solution contains item copies of items in $I$, but at most one item copy whose corresponding item is in $\tII$. To be more precise, the multiset $V$ uses only one item $a \in \tII$ with a multiplicity $x_a > 0$. We have $x_a \leq 1$, but $x_{a''} = 0$ for all other $a'' \in \tII$, i.e.\ $\sum_{a' \in \tII} x_{a'} \leq 1$.
Similarly, ``the solution $V$ uses at most two items in $\tII$'' means that there are only two item copies whose corresponding item(s) are in $\tII$: we have $\sum_{a' \in \tII} x_{a'} \leq 2$. 

The interchangeable use of ``item'' and ``item copy'' allows for shorter sentences. Moreover, it is based upon 0-1 KP where ``item'' and ``item copy'' are in fact identical.
\end{remark}
}
\subsection{A First Approximation}
We present a simple approximation algorithm for $\OPT(I)$. Take the most efficient item $\ame := \mathrm{arg} \max_{a \in I} \frac{p(a)}{s(a)}$. Fill the knapsack with as many copies of $\ame$ as possible, i.e.\ take $\unterks{\frac{c}{s(\ame)}} \stackrel{c = 1}{=} \unterks{\frac{1}{s(\ame)}}$ copies of $\ame$. Then the following holds:
\begin{theorem} \label{thm:1-2_approximation_P_0}
We have $P_0 := p(\ame) \cdot \unterks{\frac{c}{s(\ame)}} \geq \frac{1}{2} \OPT(I)$. The value $P_0$ can be found in time $\Oh{n}$ and space $\Oh{1}$.
\end{theorem}
\begin{proof}
Suppose first that $\ame$ can greedily fill the knapsack completely. Then $p(\ame) \cdot \unterks{\frac{c}{s(\ame)}} = \OPT(I)$. Otherwise, one additional item $\ame$ exceeds the capacity $c$. Then $p(\ame) \cdot \unterks{\frac{c}{s(\ame)}} + p(\ame) \geq \OPT(I)$. If $p(\ame) \leq \frac{1}{2} \OPT(I)$, then $p(\ame) \cdot \unterks{\frac{c}{s(\ame)}} \geq \OPT(I)  - p(\ame) \geq \frac{1}{2} \OPT(I)$, and the theorem follows. Otherwise $p(\ame) \cdot \unterks{\frac{c}{s(\ame)}} \geq p(\ame) \geq \frac{1}{2} \OPT(I)$, which also proves the theorem.

To determine $P_0$, we only have to check all items (which can be done in $\Oh{n}$) and to save the most efficient item (which only needs time in $\Oh{1}$). 

(The proof is taken from \cites[p.\ 232]{Kellerer2004}{Lawler1979})
\end{proof}
\begin{assumption}
From now on, we assume without loss of generality that $\eps \leq \frac{1}{4}$ and $\eps = \frac{1}{2^{\kappa-1}}$ for $\kappa \in \natur$. Otherwise, we replace $\eps$ by the corresponding $\frac{1}{2^{\kappa-1}}$ such that $\frac{1}{2^{\kappa-1}} \leq \eps < \frac{1}{2^{\kappa-2}}$. Note that $\log_2 (\frac{2}{\eps}) = \kappa$ holds.
\end{assumption}
Similar to Lawler \cite{Lawler1979}, we introduce the threshold $T$ and a constant $K$:
\begin{equation}
T:= \frac{1}{2} \eps P_0 = \frac{1}{2} \frac{1}{2^{\kappa-1}} P_0 
\label{eq:definition_T}
\end{equation}
and
\begin{equation}
K := \frac{1}{4} \frac{1}{\log_2(\frac{2}{\eps}) +1} \eps T = \frac{1}{4} \frac{1}{\kappa + 1} \fraczk T = \frac{1}{8} \frac{1}{\kappa + 1} \klammer{\fraczk}^2 P_0\enspace.
\label{eq:definition_K}
\end{equation}
We will see later that these values are indeed the right choice for the algorithm. (A derivation of these values is presented in \cite[Subsection 5.8.1]{Kraft2015}.)
\section{Reducing the Items} \label{sec:reducing_items}
We first partition the items into large(-profit) and small(-profit) items, and only keep the most efficient small item:
	\[I_L := \menge{a \in I \ | \ p(a) \geq T }, \ I_S := I \setminus I_L, \ \text{ and } \aeff := \mathrm{arg}\max \menge{\frac{p(a)}{s(a)} \ \Big| \ p(a) < T} \enspace.
\]
\begin{theorem}\label{thm:construction_I-L_and_a-eff}
The sets $I_L, I_S$ and the item $\aeff$ can be found in time $\Oh{n}$ and space $\Oh{n}$. This is also the space needed to save $I_L$.
\end{theorem}
\begin{proof}
Obvious.
\end{proof}
Similar to Lawler, we now reduce the item set $I_L$. Note that we have $\OPT(I) \leq 2 P_0$ according to Theorem \ref{thm:1-2_approximation_P_0}, and one item cannot have a profit larger than $\OPT(I) \leq 2 P_0$. Hence, the large item profits are in the interval $[T, 2 P_0]$. We partition this interval into
\begin{equation} \label{eq:def:Lk}
	\Lk := [2^k T, 2^{k+1} T) \ \text{ for } k \in \menge{ 0, \ldots, \kappa + 1}\enspace.
\end{equation}
Note that 
\begin{IEEEeqnarray*}{rCl}
\Le{\kappa} & = & \intervLr{2^{\kappa} T, \; 2^{\kappa+1} T} = \intervLr{2^\kappa \frac{1}{2} \fraczk P_0, \; 2^{\kappa+1} \frac{1}{2} \fraczk P_0} = \intervLr{P_0, 2 P_0}\enspace.
\end{IEEEeqnarray*}
For convenience, we directly set $\Le{\kappa+1} := \menges{2 P_0}$.

We further split the $\Lk$ into disjoint sub-intervals, each of length $2^k K$:
\begin{equation}
\Lkg := \intervLr{ 2^k T + \gamma \cdot 2^k K, \; 2^k T + (\gamma + 1) 2^k K } \ \text{ for } \ \gamma \in \menge{ 0, \ldots, 2^{\kappa+1} (\kappa + 1) - 1}\enspace.
\label{eq:def:Lkg}
\end{equation}
Note that indeed $\Lk = \bigcup_\gamma \Lkg$ holds because
\begin{IEEEeqnarray*}{rCl}
 2^k T + (\gamma + 1) 2^k K|_{\gamma = 2^{\kappa + 1}(\kappa + 1) - 1} & = & 2^k T + 2^{\kappa+1} (\kappa + 1)2^k K\\
& \stackrel{\eqref{eq:definition_K}}{=} & 2^k T + 2^{\kappa+1} (\kappa + 1) 2^k \frac{1}{4} \frac{1}{\kappa + 1} \fraczk T \\
& = & 2^k T + 2^k T = 2^{k+1} T\enspace.
\end{IEEEeqnarray*}
Similar to above, we set $\Lege{\kappa+1}{0} := \menges{2 P_0}$.

The idea is to keep only the smallest item $a$ for every profit interval $\Lkg$. We will see that these items are sufficient to determine an approximate solution.
\begin{definition}
For an item $a$ with $p(a) \geq T$, let $k(a) \in \natur$ be the interval such that $p(a) \in \Le{k(a)}$ and $\gamma(a) \in \natur$ be the sub-interval such that $p(a) \in \Lege{k(a)}{\gamma(a)}$. Let $\akg$ be the smallest item for the profit interval $\Lkg$, i.e.\
	\[\akg := \mathrm{arg}\min \menge{ s(a) \ | \ a \in I_L \textrm{ and } p(a) \in \Lkg } \textrm{ for all } k \textrm{ and } \gamma\enspace.
\]
\end{definition}
Algorithm \ref{alg:find_a-k-g} shows the algorithm to determine the $\akg$. They form the reduced set of large items 
	\[\ILred := \bigcup_k \bigcup_\gamma \menges{\akg}\enspace.
\]
\begin{algorithm}

\For{$k = 0, \ldots, \kappa$}{
 \For{$\gamma = 0, \ldots, 2^{\kappa+1}(\kappa+1) - 1$}{
		$\akg := \emptyset$\;}
		}
$\aege{\kappa+1}{0} := \emptyset$\;		
\For{$a \in I_L$}{
  Determine $(k(a),\gamma(a))$\;
	\If{$s(\aege{k(a)}{\gamma(a)}) > s(a)$ or $\aege{k(a)}{\gamma(a)} = \emptyset$}{$\aege{k(a)}{\gamma(a)} := a$\;}
	}
\KwOut{$\ILred := \bigcup_k \bigcup_\gamma \menges{\akg}$}
\caption{The algorithm to determine the $\akg$.}%
\label{alg:find_a-k-g}%
\end{algorithm}
As in \cite{Lawler1979}, we now prove that $\ILred$ is sufficient for an approximation.
\begin{lemma}\label{lemma:solution_quality_I-L_a-eff}
Let $0 \leq v \leq c=1$. Then 
	\[\OPTvv{\menges{\aeff}}{c-v} \geq \OPTvv{I_S}{c-v} - T
\]
 and 
	\begin{align*}
	\OPTv{\ILred} 
	&\geq \klammer{1-\frac{\eps}{4} \fraclogeps } \OPTv{I_L}\enspace.
\end{align*}
\end{lemma}
\begin{proof}
For the first inequality, there are two possibilities: either copies of $\aeff$ can be taken such that the entire capacity $c-v$ is used. Then obviously $\OPTvv{\menges{\aeff}}{c-v} = \OPTvv{I_S}{c-v}$ holds. Otherwise, we have similar to the proof of Theorem \ref{thm:1-2_approximation_P_0} that $\OPTvvs{\menges{\aeff}}{c-v} + p(\aeff) = \unterks{\frac{c-v}{s(\aeff)}} \cdot p(\aeff) + p(\aeff) \geq \OPTvv{I_S}{c-v}$. Thus, $\OPTvvs{\menges{\aeff}}{c-v} \geq \OPTvvs{I_S}{c-v} - p(\aeff) \geq \OPTvvs{I_S}{c-v}-T$. The first inequality follows.

For the second inequality, take an optimal solution $(x_a)_{a \in I}$ such that $\OPTv{I_L} = \sum_{a \in I_L} p(a) x_a$. Replace now every item $a$ by its counterpart $\aege{k(a)}{\gamma(a)}$ in $\ILred$. Obviously, the solution stays feasible, i.e.\ the volume $v$ will not be exceeded, because an item may only be replaced by a smaller one. This solution has the total profit $\sum_{a \in I_L} p(\aege{k(a)}{\gamma(a)}) x_a$. Moreover, we have 
	\begin{IEEEeqnarray*}{rCl}
	p(\aege{k(a)}{\gamma(a)}) & \geq & p(a) - 2^{k(a)} K  \stackrel{\eqref{eq:definition_K}}{=} p(a) - \frac{1}{4} \frac{1}{\kappa+1} \fraczk 2^{k(a)} T \\
&	\stackrel{p(a) \geq 2^{k(a)} T}{\geq}& p(a) - \klammer{\frac{1}{4} \frac{1}{\kappa+1}\fraczk} p(a) 
= p(a) \cdot \klammer{1-\frac{1}{4} \frac{1}{\kappa+1}\fraczk} \IEEEeqnarraynumspace \IEEEyesnumber \label{eq:proof:distance_p-a_p-aege}
\end{IEEEeqnarray*}
by the definition of the $\Lkg$. We get
\begin{IEEEeqnarray*}{rCl}
\OPTv{\ILred} &\geq& \sum_{a \in I_L} p(\aege{k(a)}{\gamma(a)}) x_a  \stackrel{\eqref{eq:proof:distance_p-a_p-aege}}{\geq}  \sum_{a \in I_L} \klammer{1-\frac{1}{4} \frac{1}{\kappa+1}\fraczk} \cdot p(a) x_a \\
 & = & \klammer{1-\frac{1}{4} \frac{1}{\kappa+1}\fraczk} \OPTv{I_L} \\
 & = & \klammer{1-\frac{\eps}{4} \fraclogeps} \OPTv{I_L}\enspace. 
\end{IEEEeqnarray*}
(The reasoning is partially taken directly from or close to the one by Lawler in \cite{Lawler1979}.)
\end{proof}
\begin{theorem}\label{thm:number_items_I-L-red_time_space_constructing_I-L-red}
The set $\ILred$ has $\Ohs{\freps \log^2 \freps}$ items.   
Algorithm \ref{alg:find_a-k-g} needs time in $\Ohs{n + \freps \log^2 \freps}$ and space in $\Ohs{\freps \log^2 \freps}$ for the construction and for saving $\ILred$.
\end{theorem}
\begin{proof}
The number of items $\akg$, including the item $\aege{\kappa+1}{0}$, is bounded by $\Ohs{(\kappa + 1) \cdot (\zke{\kappa+1} (\kappa+1) - 1 +1 ) } = \Ohs{\log \freps \cdot (\freps \log \freps)} = \Ohs{\freps \log^2 \freps}$.
The space needed is asymptotically bounded by the space required to save the $\akg$. Finally, the running time is obviously bounded by $\Ohs{n + \freps \log^2 \freps}$: the values $k(a)$ and $\gamma(a)$ can be found in $\Ohs{1}$ because we assume that the logarithm can be determined in $\Oh{1}$.
\end{proof}
\begin{remark} \label{remark:no_item_with_profit_2P0}
If there is one item $a$ with the profit $p(a) = 2 P_0$, i.e.\ whose profit attains the upper bound, one optimum solution obviously consists of this single item. During the partition of $I$ into $I_L$ and $I_S$, it can easily be checked whether such an item is contained in $I$. Since the algorithm can directly stop if this is the case, we will from now on assume without loss of generality that such an item does not exist and that $\aege{\kappa+1}{0} = \emptyset$.
\end{remark}

%% file: simplify_structure.tex
\section{A Simplified Solution Structure}
In this section, we will transform $\ILred$ into a new instance $\tI$ whose optimum $\OPTvs{\tI}$ is only slightly smaller than $\OPTvs{\ILred}$ and where the corresponding solution has a special structure. This new transformation will allow us later to faster construct the approximate solution.
First, we define
	\[\Ik := \menge{a \in \ILred \ \big| \ p(a) \in \Lk} = \menge{a \in \ILred \ \big| \ p(a) \in \intervLr{2^{k} T, 2^{k+1}T}}\enspace.
\]
Note that the items are already partitioned into the $\Ik$ because of the way $\ILred$ has been constructed.
\begin{definition}
Let $a_1, a_2$ be two knapsack items with $s(a_1) + s(a_2) \leq c$. The gluing operation $\oplus$ combines them into a new item $a_1 \oplus a_2$ with $p(a_1 \oplus a_2) = p(a_1) + p(a_2)$ and $s(a_1 \oplus a_2) = s(a_1) + s(a_2)$.
\end{definition}
Thus, the gluing operation is only defined on pairs of items whose combined size does not exceed $c$.

The basic idea for the new instance $\tI$ is as follows: we first set $\tIe{0} := \Ie{0}$. Then, we construct $a_1 \oplus a_2$ for all $a_1, a_2 \in \tIe{0}$ (which also includes the case $a_1 = a_2$), which yields the item set $\ttIe{1} := \menges{a_1 \oplus a_2 \ | \ a_1, a_2 \in \tIe{0}}$. Note that $p(a_1 \oplus a_2) \in [2T, 4T) = \Le{1}$. For every profit interval $\Lege{1}{\gamma}$, we keep only the item of smallest size in $\Ie{1} \cup \ttIe{1}$, which yields the item set $\tIe{1}$. This procedure is iterated for $k = 1, \ldots, \kappa-1$: the set $\tIe{k}$ contains the items with a profit in $[2^{k}T, 2^{k+1} T) = \Lk$ (see Fig.\@ \ref{fig:construction_tIk}\subref{fig:construction_tIk_a}). Gluing like above yields the item set $\ttIe{k+1}$ with profits in $[2^{k+1}T, 2^{k+2} T) = \Le{k+1}$ (see Fig.\@ \ref{fig:construction_tIk}\subref{fig:construction_tIk_b}). By taking again the smallest item in $\ttIe{k+1} \cup \Ie{k+1}$ for every $\Leg{k+1}$, the set $\tIe{k+1}$ is derived (see Fig.\@ \ref{fig:construction_tIk}\subref{fig:construction_tIk_c}). The item in $\tIe{k}$ with a profit in $\Leg{k}$ is denoted by $\taeg{k}$ for every $k$ and $\gamma$. 

We finish when $\tIe{\kappa}$ has been constructed. We are in the case where $\Ie{\kappa+1} = \emptyset$, i.e.\ $\aege{\kappa+1}{0} = \emptyset$, and it is explained at the beginning of Section \ref{sec:dynamic_programming} that it is not necessary to construct $\tIe{\kappa+1}$ from the items in $\tIe{\kappa}$. Hence, we also have $\taege{\kappa+1}{0} = \emptyset$.

Note that we may glue items together that already consist of glued items. For backtracking, we save for every $\takg$ which two items in $\tIe{k-1}$ have formed it or whether $\takg$ has already been an item in $\Ik$. Algorithm \ref{alg:construction_tIk} presents one way to construct the sets $\tIk$.
\begin{remark}\label{remark:items_takg_represent_original_items}
One item $\takg$ is in fact the combination of several items in $\ILred$. The profit and size of $\takg$ is equal to the total profit and size of these items. The $\takg$ represent feasible item combinations because an arbitrary number of item copies can be taken in UKP.
\end{remark}
\begin{algorithm}
\For{$k = 0, \ldots, \kappa$}{
	\For{$\gamma = 0, \ldots, \zke{\kappa+1} (\kappa + 1) - 1 $}{
		$\takg := \akg$\; $\backtrack(\takg) := \akg$\;}
	}

$\tIe{0} := \Ie{0}$\;
\For{$k = 0, \ldots, \kappa-1$}{
	\For{$\gamma = 0, \ldots, \zke{\kappa+1} (\kappa + 1) - 1 $}{
		\For{$\gamma' = \gamma, \ldots, \zke{\kappa+1} (\kappa + 1) - 1$}{
			\If{$s(\takg) + s(\taege{k}{\gamma'}) \leq c$}			
			{$\ta := \takg \oplus \taege{k}{\gamma'}$\;
			\If{$s(\ta) < s(\taege{k+1}{\gamma(\ta)})$ or $\taege{k+1}{\gamma(\ta)} = \emptyset$}{
				$\taege{k+1}{\gamma(\ta)} := \ta$\;
				$\backtrack(\taege{k+1}{\gamma(\ta)}) := (\takg, \taege{k}{\gamma'})$\;}
			}
		}
	}
	$\tIe{k+1} := \menge{\taege{k+1}{0}, \ldots, \taege{k+1}{\zke{\kappa+1} (\kappa + 1) - 1}}$\;
}
\caption{The construction of the item sets $\tIk$.}
\label{alg:construction_tIk}
\end{algorithm}
\begin{figure}%
\subfloat[The items in $\tIk$ and $\Ie{k+1}$\label{fig:construction_tIk_a}. The height of every item $a$ corresponds to its size $s(a)$ while its position on the axis corresponds to its profit $p(a)$. The axis is partitioned into the profit sub-intervals $\Lkg = [2^k T + \gamma 2^k K, 2^k T + (\gamma+1) 2^k K)$.]{\includegraphics[width=\columnwidth]{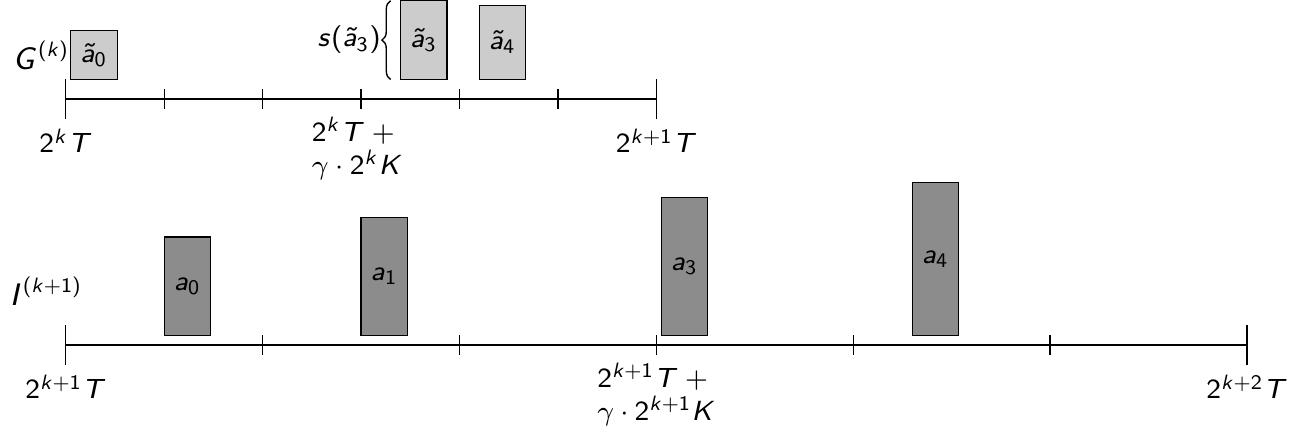}}

\subfloat[The set $\Ie{k+1}$ together with the newly constructed items in $\ttIe{k+1}$\label{fig:construction_tIk_b}]{\includegraphics[width=\columnwidth]{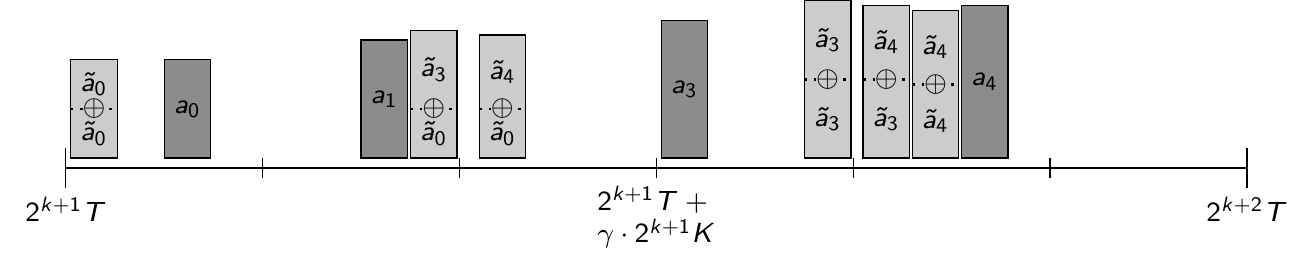}}

\subfloat[The new set $\tIe{k+1}$ after keeping only the smallest item with a profit in $\Leg{k+1}$. For instance, $\ta_4 \oplus \ta_4$ is kept because it is the smallest item in its profit sub-interval $\Leg{k+1}$.\label{fig:construction_tIk_c}]{\includegraphics[width=\columnwidth]{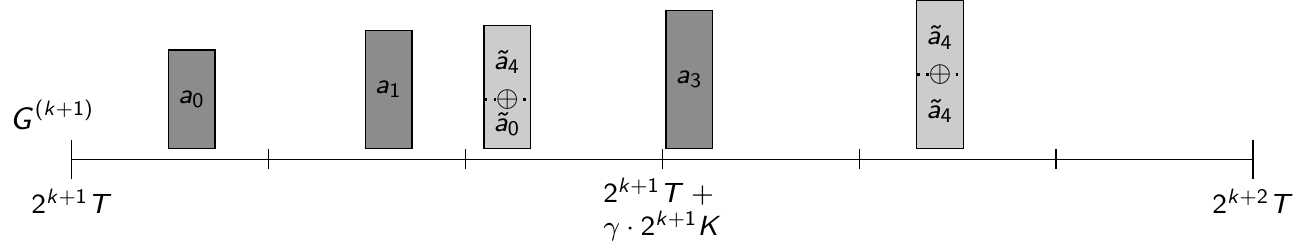}}
\caption{Principle of deriving $\tIe{k+1}$ from $\tIk$ and $\Ie{k+1}$}%
\label{fig:construction_tIk}%
\end{figure}

The item set 
	\[\tI := \bigcup_{k=0}^{\kappa} \tIk
\]
has for every $0 \leq v \leq c$ a solution near the original optimum $\OPTv{\ILred}$ as shown below in Theorem \ref{thm:solution_quality_tIk}. It is additionally proved that at most one item of every $\tIk$ for $k \in \menges{ 0, \ldots, \kappa - 1}$ is needed. First, we introduce a definition for the proof.
\begin{definition} \label{def:structured_solution_OPTk}
Let $I'$ be a set of knapsack items with $p(a) \geq T$ for every $a \in I'$. For a knapsack volume $v \leq c$ and $k_0 \in \menges{0, \ldots, \kappa}$, a solution is \emph{structured for $k = k_0$} if it fits into $v$ and uses for every $k \in \menges{ 0, \ldots, k_0}$ at most one item copy with a profit in $\Lk$. We denote by $\OPTkkv{k_0}{I'}$ the corresponding optimum profit.
\end{definition}
For instance, the solution for
	\[\OPTkkv{k_0}{\tIe{0} \cup \ldots \cup \tIe{k_0} \cup \tIe{k_0+1} \cup \Ie{k_0+2} \cup \ldots \cup \Ie{\kappa}} 
\]
fits into the volume $v$, and it uses only one item from every $\tIe{k}$ for $k \in \menges{ 0, \ldots, k_0}$. It may however use an arbitrary number of item copies e.g.\ in $\tIe{k_0+1}$ or $\Ie{k_0+2}$.

\begin{theorem} \label{thm:solution_quality_tIk}
For $v \leq c$ and $k_0 \in \menges{ 0, \ldots, \kappa-1}$, we have 
	\[\OPTkkv{k_0}{\bigcup_{k=0}^{k_0+1} \tIe{k} \cup \bigcup_{k = k_0+2}^\kappa \Ie{k}} \geq \klammer{ 1- \frac{\eps}{4} \fraclogeps }^{k_0+1} \OPTv{\ILred}\enspace.
\]
\end{theorem}
\begin{proof}
The proof idea is quite simple: we iteratively replace the items in $\Ie{k_0+1}$ by their counterpart in $\tIe{k_0+1}$ and also replace every pair of item copies in $\tIe{k_0}$ by the counterpart in $\tIe{k_0+1}$. This directly follows the way to construct the item sets $\tIe{k}$ presented in Algorithm \ref{alg:construction_tIk}.

Formally, the statement is proved by induction over $k_0$. Let $k_0 = 0$. 
Take an optimum solution to $\tIe{0} \cup \Ie{1} \cup \ldots \cup \Ie{\kappa} = \Ie{0} \cup \Ie{1} \cup \ldots \cup \Ie{\kappa} = \ILred$. For ease of notation, we directly write each item as often as it appears in the solution. We have three sub-sequences:
\begin{itemize}
	\item Let $\ba_1, \ldots, \ba_\eta$ ($\eta \in \natur$) be the items from $\tIe{0} = \Ie{0}$ in the optimal solution for $\OPTvs{\ILred}$. We assume that $\eta$ is odd (the case where $\eta$ is even is easier and handled below.)
	\item Let $\ba_{\eta + 1}, \ldots, \ba_{\eta + \xi}$ ($\xi \in \natur$) be the items from $\Ie{1}$ in the optimal solution for $\OPTvs{\ILred}$.
	\item Let $\bas_1, \ldots, \bas_\lambda$ ($\lambda \in \natur$) be the remaining items from $\Ie{2} \cup \ldots \cup \Ie{\kappa}$ in the optimal solution for $\OPTvs{\ILred}$. This set is denoted by $\Lambda$. As defined above, the total profit of these items is written as $p(\Lambda)$.
\end{itemize}
Figure \ref{fig:solution_quality_tIk}\subref{fig:solution_quality_tIk_a} illustrates the packing. (Figure \ref{fig:solution_quality_tIk} shows the case for general $k$.) We have
\begin{equation} \label{eq:proof:profit_before_gluing}
 \OPTv{\tIe{0} \cup \Ie{1} \cup \ldots \cup \Ie{\kappa}} = \sum_{i = 1}^{\eta} p(\ba_i) + \sum_{j = \eta + 1}^{\eta + \xi} p(\ba_j) + p(\Lambda)\enspace.
\end{equation}

In the first step, every pair of items $\ba_{2i-1}$ and $\ba_{2i}$ from $\tIe{0}$ for $i \in \menges{ 1, \ldots, \unterks{\frac{\eta}{2}}}$ is replaced by $\ba_{2i-1}\oplus \ba_{2i} \in \ttIe{1}$ (see Fig.\@ \ref{fig:solution_quality_tIk}\subref{fig:solution_quality_tIk_b}). In the second step, every item $\ba_{2i-1}\oplus \ba_{2i}$ is again replaced by the corresponding item $\taege{1}{\gamma(\ba_{2i-1}\oplus \ba_{2i})} =: \taege{1}{\rho(i)}$ in $\tIe{1}$ (for $i \in \menges{ 1, \ldots, \unterks{\frac{\eta}{2}} }$). Only item $\ba_\eta$ remains unchanged. Moreover, $\ba_j$ from $\Ie{1}$ is replaced by the corresponding $\taege{1}{\gamma(\ba_j)} =: \taege{1}{\rho(j)}$ for $j \in \menges{ \eta+1, \ldots, \eta + \xi }$ (see Fig.\@ \ref{fig:solution_quality_tIk}\subref{fig:solution_quality_tIk_c}). Note that this new solution is indeed feasible because the replacing items $\taege{1}{\gamma}$ are at most as large as the original ones. Moreover, the corresponding items $\taege{1}{\rho(i)}$ and $\taege{1}{\rho(j)}$ must exist by the construction of $\tIe{1}$. Thus, we have a (feasible) solution that consists of the item $\ba_\eta \in \tIe{0}$, the items $\taege{1}{\rho(i)}$ and $\taege{1}{\rho(j)}$ in $\tIe{1}$, and the remaining items $\bas_1, \ldots, \bas_\lambda$ in $\Ie{2}, \ldots, \Ie{\kappa}$: this solution respects the structure of $\OPTkkvs{k_0}{\cdot}$ for $k_0 = 0$. (If $\eta$ is even, no item in $\tIe{0}$ is used.)

Let now $\ba$ be an item $\ba_{2i-1}\oplus \ba_{2i}$ or $\ba_j$. It can be proved as for Inequality \eqref{eq:proof:distance_p-a_p-aege} that
\begin{IEEEeqnarray}{rCl} \label{eq:proof:distance_p-ba_p_baege}
 p(\taege{1}{\gamma(\ba)}) & \geq & \klammer{1- \frac{\eps}{4} \fraclogeps} p(\ba) \enspace.
\end{IEEEeqnarray}
Thus, we have
\begin{IEEEeqnarray*}{rCl}
\IEEEeqnarraymulticol{3}{l}{ \OPTkkv{0}{\tIe{0} \cup \tIe{1} \cup \Ie{2} \cup \ldots \cup \Ie{\kappa}}  \geq  p(\ba_\eta) + \sum_{i = 1}^{\unterks{\frac{\eta}{2}}} p(\taege{1}{\rho(i)}) + \sum_{j = \eta + 1}^{\eta + \xi} p(\taege{1}{\rho(j)}) + p(\Lambda)} \\ \qquad
& \stackrel{\eqref{eq:proof:distance_p-ba_p_baege}}{\geq} & p(\ba_\eta) + \klammer{1- \frac{\eps}{4} \fraclogeps} \sum_{i = 1}^{\unterks{\frac{\eta}{2}}} p(\ba_{2i-1}\oplus \ba_{2i})\\
& & + \klammer{1- \frac{\eps}{4} \fraclogeps} \sum_{j=\eta+1}^{\eta + \xi} p(\ba_j) + p(\Lambda)\\
& \geq & \klammer{1- \frac{\eps}{4} \fraclogeps} \klammer{\sum_{i = 1}^{\eta} p(\ba_i) + \sum_{j = \eta + 1}^{\eta + \xi} p(\ba_j) + p(\Lambda)}\\
& \stackrel{\eqref{eq:proof:profit_before_gluing}}{=} & \klammer{1- \frac{\eps}{4} \fraclogeps} \OPTv{\tIe{0} \cup \Ie{1} \cup \ldots \cup \Ie{\kappa}}\\
& = & \klammer{1- \frac{\eps}{4} \fraclogeps} \OPTv{\ILred}\enspace.
\end{IEEEeqnarray*}
The statement for $k_0 = 1, \ldots, \kappa-1$ now follows by induction. The proof is almost identical to the case $k_0 = 0$ above, the only difference is that there are additionally the items in $\tIe{0}, \ldots, \tIe{k_0-1}$ that remain unchanged like the items $\Ie{k_0 + 2}, \ldots, \Ie{\kappa}$. Only the items in $\tIe{k_0}$ and $\Ie{k_0}$ are replaced.
\end{proof}
\begin{figure}
\centering
\subfloat[The current solution to $\tIe{0} \cup \cdots \cup \tIe{k} \cup \Ie{k+1} \cup \cdots \cup \Ie{\kappa}$. The structure of $\OPTkkvs{k-1}{\cdot}$ is respected, i.e.\ at most one item from every $\tIe{0}, \ldots, \tIe{k-1}$ is used.\label{fig:solution_quality_tIk_a}]{\includegraphics[scale=1.4]{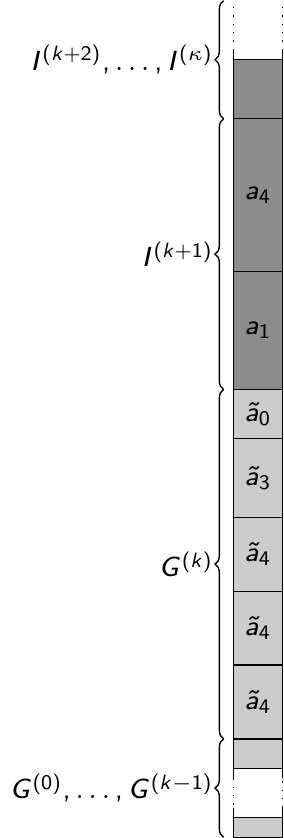}}
\hspace{1cm}
\subfloat[The items in $\tIk$ are pairwise glued together with the possible exception of one item.\label{fig:solution_quality_tIk_b}]{\includegraphics[scale=1.4]{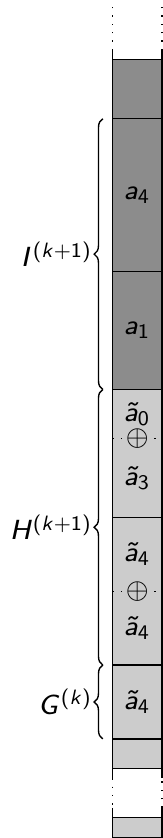}}
\hspace{1cm}
\subfloat[The items in $\ttIe{k+1} \cup \Ie{k+1}$ are replaced by their counterparts in $\tIe{k+1}$. Now, at most one item in $\tIk$ is part of the solution, and the structure for $\OPTkvs{\cdot}$ is respected.\label{fig:solution_quality_tIk_c}
]
{\includegraphics[scale=1.4]{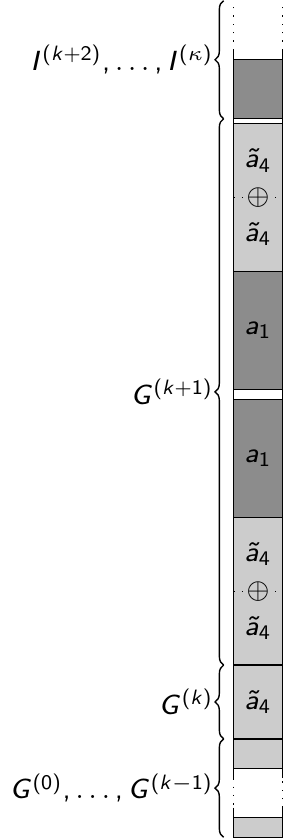}}
\caption{The principle of the proof for Theorem \ref{thm:solution_quality_tIk}}%
\label{fig:solution_quality_tIk}%
\end{figure}

\begin{lemma} \label{lemma:estimates_I-L_to_tI}
$\OPT(\tI \cup \menges{\aeff}) \leq \OPT(\ILred \cup I_S) \leq \OPT(I_L \cup I_S) = \OPT(I) \leq 2 P_0$ holds.
\end{lemma}
\begin{proof}
$\tI$ consists of items in $\ILred$ or of items that can be obtained by gluing several items in $\ILred$ together. Every combination of items in $\tI$ can therefore be represented by items in $\ILred$. Moreover, we have $\aeff \in I_S$. The first inequality follows. Since $\ILred \subseteq I_L$, the second inequality is obvious. The last inequality follows from Theorem \ref{thm:1-2_approximation_P_0}.
\end{proof}
Up to now, we have (only) reduced the original item set $I$ to $\tI \cup \menges{\aeff}$.
\begin{lemma} \label{lemma:two_cases}
Assume as mentioned in Remark \ref{remark:no_item_with_profit_2P0} that $\aege{\kappa+1}{0} = \emptyset$. Consider the optimum structured solutions to $\tI \cup \menges{\aeff}$ for $k_0 = \kappa-1$ (see Definition \ref{def:structured_solution_OPTk}). This means that at most one item is used from every $\tIk$ for $k \in \menges{ 0, \ldots, \kappa-1}$. (The item $\aeff$ has a profit $p(\aeff) < T$ such that it does not have to satisfy any structural conditions.) Then there are two possible cases:
\begin{itemize}
	\item One solution uses (at least) two items in $\tIe{\kappa}$. This is the case if and only if the optimum for $\tI \cup \menges{\aeff}$ is $2P_0$, and the solution consists of two item copies of the item $\taege{\kappa}{0}$ with $p(\taege{\kappa}{0}) = P_0$.
	\item Every solution uses at most one item in $\tIe{\kappa}$. Then, $\OPTkkvvs{\kappa-1}{\tI}{v'} = \OPTkkvvs{\kappa}{\tI}{v'}$ holds for all values $0 \leq v' \leq c$, and there is a value $0 \leq v \leq c$ such that
	\begin{IEEEeqnarray*}{rCl}
		\IEEEeqnarraymulticol{3}{l}{\OPTkkv{\kappa}{\tI} +  \OPTvv{\menge{\aeff}}{c-v} = \OPTkkv{\kappa-1}{\tI} + \OPTvv{\menge{\aeff}}{c-v}} \\ \quad
			&\geq & \klammer{1-\frac{\eps}{4}\fraclogeps}^{\kappa+1} \OPT(I) - T\enspace.
	\end{IEEEeqnarray*}
	Moreover, $\OPTkkvs{\kappa}{\tI}$ uses at least one item in $\tIe{\kappa-2} \cup \tIe{\kappa-1} \cup \tIe{\kappa}$, and/or we have $\OPTvv{\menge{\aeff}}{c-v} \geq \frac{1}{4} P_0$.
\end{itemize}
\end{lemma}
\begin{proof}
Note that $\ILred$ does not contain any item with the profit $2 P_0$ (see Remark \ref{remark:no_item_with_profit_2P0}). By construction, this is still the case for $\tI$. Suppose now that one solution to $\tI \cup \menges{\aeff}$ uses more than one item in $\tIe{\kappa}$. Since items in $\tIe{\kappa}$ have profits in $[P_0, 2P_0)$, only two copies of the item $\taege{\kappa}{0}$ can be used, and we have $p(\taege{\kappa}{0}) = P_0$. In fact, $2 P_0$ is the maximum possible profit because $\OPT(\tI \cup \menges{\aeff}) \leq \OPT(I) \leq 2 P_0$ holds as we have seen in 
Lemma \ref{lemma:estimates_I-L_to_tI}. Thus, the ``only if'' direction has been proved. The ``if''-direction is obvious.

Suppose now that every structured solution to $\tI \cup \menges{\aeff}$ for $k_0 = \kappa-1$ uses at most one item in $\tIe{\kappa}$. Thus, $\OPTkkvvs{\kappa-1}{\tI}{v'} = \OPTkkvvs{\kappa}{\tI}{v'}$ holds for all $0 \leq v' \leq c$. 

Let $v \leq c$ now be the volume the large items $I_L$ occupy in an optimum solution to $I$. Then obviously $\OPT(I) = \OPTv{I_L} + \OPTvv{I_S}{c-v}$ holds. We have the following inequality:
\begin{IEEEeqnarray*}{rCl}
\IEEEeqnarraymulticol{3}{l}{  \OPTkkv{\kappa}{\tI}  +  \OPTvv{\menge{\aeff}}{c-v}  = \OPTkkv{\kappa-1}{\tI}  +  \ \OPTvv{\menge{\aeff}}{c-v}  } \\ \quad
&\stackrel{\text{Thm.~\ref{thm:solution_quality_tIk}}}{\geq}& \klammer{1-\frac{\eps}{4}\fraclogeps}^{\kappa} \OPTv{\ILred} + \OPTvv{\menge{\aeff}}{c-v}\\
&\stackrel{\text{Lem.~\ref{lemma:solution_quality_I-L_a-eff}}}{\geq}& \klammer{1-\frac{\eps}{4}\fraclogeps}^{\kappa+1} \OPTv{I_L} + \OPTvv{I_S}{c-v} - T\\
&\geq& \klammer{1-\frac{\eps}{4}\fraclogeps}^{\kappa+1} \klammer{\OPTv{I_L} + \OPTvv{I_S}{c-v}} - T\\
& = &  \klammer{1-\frac{\eps}{4}\fraclogeps}^{\kappa+1} \OPT(I) - T\enspace. \label{eq:main_estimate_combination} \IEEEyesnumber
\end{IEEEeqnarray*}

For the final property, suppose that no item in $\tIe{\kappa-2} \cup \tIe{\kappa-1} \cup \tIe{\kappa}$ is used in a solution for $\OPTkkvs{\kappa}{\tI}$. Then we have
	\[\OPTkkv{\kappa}{\tI} \leq \sum_{k=0}^{\kappa-3} \max\menge{p(a) \ | \ a \in \tIk} \leq \sum_{k=0}^{\kappa-3}2 \cdot 2^k T < 2^{\kappa-1} T \stackrel{\eqref{eq:definition_T}}{=} \frac{1}{2} P_0\enspace.
\]
On the other hand, Inequality \eqref{eq:main_estimate_combination} together with $(1-\delta)^k \geq (1-k \cdot \delta)$ for $\delta < 1$ yields
\begin{IEEEeqnarray*}{rCl}
\IEEEeqnarraymulticol{3}{l}{\OPTkkv{\kappa}{\tI} + \OPTvv{\menge{\aeff}}{c-v} \geq  \klammer{1-\frac{\eps}{4} \efraclogeps{\kappa+1}} \OPT(I) - T }\\ \quad
& \stackrel{\eqref{eq:definition_T}}{=}& \klammer{1-\frac{\eps}{4}}\OPT(I) - \frac{1}{2} \eps P_0 
 \geq \klammer{1-\frac{\eps}{4}} \OPT(I) - \frac{1}{2} \eps \OPT(I) \\
&\stackrel{\eps\leq \nicefrac{1}{4}}{\geq}& \frac{3}{4} \OPT(I) \geq \frac{3}{4} P_0\enspace.
\end{IEEEeqnarray*}
Hence, $\OPTvvs{\menge{\aeff}}{c-v} \geq \frac{1}{4} P_0$ holds. The final property of the second case follows.
\end{proof}
\begin{definition} \label{def:aeffc}
Take $\oberks{\frac{\nicefrac{P_0}{4}}{p(\aeff)}}$ items $\aeff$. If their total size is at most $c$, they are glued together to $\aeffc$. 
\end{definition}
Obviously, $\aeffc$ consists of the smallest number of items $\aeff$ whose total profit is at least $\frac{P_0}{4}$. Moreover, $\aeffc$ is a large item.
\begin{definition} \label{def:structured_solution_OPT-S}
Take a knapsack volume $v \leq c$. Consider the following solutions to $\tI\cup\menges{\aeffc}$ of size at most $v$:
\begin{itemize}
	\item They are structured for $k = \kappa$, i.e.\ they use for every $k \in \menges{0, \ldots, \kappa}$ at most one item in $\tIk$.
	\item They additionally use the item $\aeffc$ at most once and at least one item $\ta \in \tIe{\kappa-2} \cup \tIe{\kappa-1} \cup \tIe{\kappa} \cup \menges{\aeffc}$.
\end{itemize}
  Hence, these solutions have a profit of at least $p(\ta) \geq \zke{\kappa-2} T = \frac{1}{4} P_0$. These special solutions are called \emph{structured solutions with a lower bound} (on the profit). 
	
	The optimal profit for such solutions of total size at most $v$ is denoted by $\OPTSvs{\tI\cup\menges{\aeffc}}$. If $v$ is too small such that such a solution does not exist, we set $\OPTSvs{\tI\cup\menges{\aeffc}} = 0$.
\end{definition}
\begin{theorem} \label{thm:structured_solution_exists}
In the second case of Lemma \ref{lemma:two_cases}, there is a value $0 \leq v \leq c$ such that
\begin{IEEEeqnarray*}{c}
	\OPTSv{\tI \cup \menge{\aeffc}} +  \OPTvv{\menge{\aeff}}{c-v} 
	\geq  \klammer{1-\frac{\eps}{4}\fraclogeps}^{\kappa+1} \OPT(I) - T\enspace.
\end{IEEEeqnarray*}
\end{theorem}
\begin{proof}
Like in the proof of Lemma \ref{lemma:two_cases}, let $v'$ be the volume the large items $I_L$ occupy in an optimum solution to $I$ so that $\OPTvvs{I_L}{v'} + \OPTvvs{I_S}{c-v'} = \OPT(I)$. Consider an optimum solution for $\OPTkkvvs{\kappa}{\tI}{v'}$ and suppose that it does not use any item in $\tIe{\kappa-2} \cup \tIe{\kappa-1} \cup \tIe{\kappa}$. Lemma \ref{lemma:two_cases} states that $\OPTvvs{\menges{\aeff}}{c-v'}$ has a profit of at least $\frac{1}{4} P_0$. Thus, a subset of the item copies of $\aeff$ can be replaced by $\aeffc$, and $c - v' \geq s(\aeffc)$. We set $v:= v' + s(\aeffc)$. Note that $\OPTSvs{\tI \cup \menges{\aeffc}} \geq \OPTkkvvs{\kappa}{\tI}{v'} + p(\aeffc)$. Moreover, $\OPTkkvvs{\kappa}{\tI}{v'} = \OPTkkvvs{\kappa-1}{\tI}{v'}$ holds because we are in the second case of Lemma \ref{lemma:two_cases}. We get the following inequalities:
\begin{IEEEeqnarray*}{rCl}
\IEEEeqnarraymulticol{3}{l}{\OPTSv{\tI \cup \menge{\aeffc}} + \OPTvv{\menge{\aeff}}{c-v}} \\ \quad
& \geq & \OPTkkvv{\kappa}{\tI}{v'} + p(\aeffc) + \OPTvv{\menge{\aeff}}{c-v'-s(\aeffc)}\\
& = & \OPTkkvv{\kappa-1}{\tI}{v'} + \OPTvv{\menge{\aeff}}{c-v'}\\
&\stackrel{\text{Thm.~\ref{thm:solution_quality_tIk}}}{\geq}& \klammer{1-\frac{\eps}{4}\fraclogeps}^{\kappa} \OPTvv{\ILred}{v'} + \OPTvv{\menge{\aeff}}{c-v'}\\
&\stackrel{\text{Lem.~\ref{lemma:solution_quality_I-L_a-eff}}}{\geq}& \klammer{1-\frac{\eps}{4}\fraclogeps}^{\kappa+1} \OPTvv{I_L}{v'} + \OPTvv{I_S}{c-v'} - T\\
&\geq& \klammer{1-\frac{\eps}{4}\fraclogeps}^{\kappa+1} \klammer{\OPTvv{I_L}{v'} + \OPTvv{I_S}{c-v'}} - T\\
& = & \klammer{1-\frac{\eps}{4}\fraclogeps}^{\kappa+1} \OPT(I) - T\enspace.
\end{IEEEeqnarray*}
Note that $\OPTvvs{\menges{\aeff}}{c-v}$ is well-defined---and therefore the entire chain of inequalities feasible---because $c-v = c-v'-s(\aeffc) \geq 0$.

Suppose now that the optimal solution uses at least one item in $\tIe{\kappa-2} \cup \tIe{\kappa-1} \cup \tIe{\kappa}$. We can then directly set $v := v'$, and the proof is similar to the first case above.

Roughly speaking, a solution in the first case of this proof satisfies the lower bound of the theorem and uses at most one item in every $\tIe{k}$, but no item in $\tIe{\kappa-2}, \tIe{\kappa-1}$ or $\tIe{\kappa}$. This implies that enough items $\aeff$ are part of the solution such that a subset of them can be replaced by $\aeffc$.
\end{proof}
So far, we have not constructed an actual solution. We have only shown in Theorem \ref{thm:structured_solution_exists} that there is a solution to $\tI \cup \menges{\aeffc} \cup \menges{\aeff}$ that is close to $\OPT(I)$ and that is a structured solution with a lower bound.

\begin{theorem}\label{thm:construction_tI_aeffc_running-time}
The cardinality of $\tIk$ is in $\Ohs{\freps \log\freps}$, i.e.\ $\tI$ has $\Ohs{\freps \log^2\freps}$ items. 
Algorithm \ref{alg:construction_tIk} constructs $\tI$ in time $\Ohs{\frepse{2} \log^3(\freps)}$ and space $\Ohs{\freps \log^2\freps}$, which also includes the space to store $\tI$ and the backtracking information. The item $\aeffc$ can be constructed in time $\Oh{1}$.
\end{theorem}
\begin{proof}
The statement for $\aeffc$ is trivial: the number of items $\aeffc$ to glue together can be determined by division.

The number of items in $\tIk$ and $\tI$ can be derived like the number of items in $\ILred$ in Theorem \ref{thm:number_items_I-L-red_time_space_constructing_I-L-red}. 
The running time of Algorithm \ref{alg:construction_tIk} is obviously dominated by the second for-loop. It is in
	\[\Oh{\kappa \cdot \klammer{2^{\kappa+1}(\kappa+1)}^2} = \Oh{\log \klammer{\freps} \cdot \klammer{\frac{1}{\eps} \log \freps}^2 } = \Oh{\frepse{2} \log^3\klammer{\freps}}\enspace.
\]
The space complexity is dominated by the space to save the $\takg$ and the backtracking information, which is again asymptotically equal to the number of items in $\tI$.
\end{proof}

%% file: dynamic_programming.tex
\section{Finding an Approximate Structured Solution by Dynamic Programming} \label{sec:dynamic_programming}
The previous section has presented three cases:
\begin{enumerate}
	\item The instance $I$ has one item of profit $2 P_0$: return this item for an optimum solution, and $\OPT(I) = 2 P_0$ (see Remark \ref{remark:no_item_with_profit_2P0}).
	\item If this is not the case, and $\tI$ has one item of profit $P_0$ and size at most $\frac{c}{2}$, two copies of this item are an optimum solution to $\tI \cup \menges{\aeff}$ (see Lemma \ref{lemma:two_cases}). Undoing the gluing returns an optimum solution with $\OPT(I) = 2 P_0$.
	\item Otherwise, there is an approximate structured solution to $\tI \cup \menges{\aeffc} \cup \menges{\aeff}$ with a lower bound (see Theorem \ref{thm:structured_solution_exists}).
\end{enumerate}
The first two cases can be easily checked, which is the reason why it has not been necessary to construct the set $\tIe{\kappa+1}$. We will from now on assume that we are in the third case: a solution uses at most one item from every $\tIk$ for $k \in \menges{ 0, \ldots, \kappa}$ as well as $\aeffc$ at most once. At the same time, at least one item $\ta \in \tIe{\kappa-2} \cup \tIe{\kappa-1} \cup \tIe{\kappa} \cup \menges{\aeffc}$ is chosen. (See Definition \ref{def:structured_solution_OPT-S}.)


We use dynamic programming to find for all $0 \leq v \leq c$ the corresponding set of large items $V \subseteq \tI \cup \menges{\aeffc}$ with $s(V) \leq v$. For convenience, let $\tIe{\kappa+1} := \menges{\aeffc}$. We introduce tuples $(p,s,k)$ similar to Lawler \cite{Lawler1979}. For profit $p$ with $0 \leq p \leq 2 P_0$ and size $0 \leq s \leq c$, the tuple $(p,s,k)$ states that there is an item set of size $s$ whose total profit is $p$. Moreover, the set has only items in $\tIk \cup \cdots \cup \tIe{\kappa+1}$ and respects the structure above.

The dynamic program is quite simple: start with the dummy tuple set $\Fke{\kappa+2}:=\menges{(0,0,\kappa+2)}$. For $k = \kappa + 1, \ldots, \kappa-2$, the tuples in $\Fk$ are recursively constructed by 
\begin{align*}
	\Fk := & \menge{(p,s,k) \ | \ (p,s,k+1) \in \Fke{k+1}} \\
	&\cup \menge{(p+p(\ta),s + s(\ta),k) \ | \ (p,s,k+1) \in \Fke{k+1}, \ta \in \tIe{k}, s+s(\ta) \leq c}\enspace.
\end{align*}
Note that $(0,0,k+1) \in \Fke{k+1}$, which guarantees that $\Fk$ also contains the tuples $(p(\ta), s(\ta), k)$ for $\ta \in \tIk$ if $k \in \menges{ \kappa+1, \ldots,\kappa-2 } $. For $k \in \menges{ \kappa-3, \ldots, 0}$, this tuple $(0,0,k+1)$ is no longer considered to form the new tuples, which guarantees that tuples of the form $(p + p(\ta), s+s(\ta), k)$ for $\ta \in \tIk$ have $p,s \neq 0$. The recursion becomes
\begin{align*}
	\Fk &:=  \menge{(p,s,k) \ | \ (p,s,k+1) \in \Fke{k+1}} \\
	&\cup \menge{(p+p(\ta),s + s(\ta),k) \ | \ (p,s,k+1) \in \Fke{k+1} \setminus\menge{(0,0,k+1)}, \ta \in \tIe{k}, s+s(\ta) \leq c}\enspace.
\end{align*}
The actual item set corresponding to $(p,s,k)$ can be reconstructed by saving backtracking information.
\begin{definition}
A tuple $(p_2,s_2,k)$ is dominated by $(p_1,s_1,k)$ if $p_2 \leq p_1$ and $s_2 \geq s_1$.
\end{definition}
As in \cite{Lawler1979}, dominated tuples $(p,s,k+1)$ are now removed from $\Fke{k+1}$ before $\Fke{k}$ is constructed. This does not affect the outcome: dominated tuples only stand for sets of items with a profit not larger and a size not smaller than non-dominated tuples. A non-dominated tuple $(p,s,k)$ is therefore optimal, i.e.\ the profit $p$ can only be obtained with items of size at least $s$ if items in $\tIk, \ldots, \tIe{\kappa+1}$ are considered. 
\begin{lemma} \label{lemma:tuple_property}
A tuple $(p,s,k) \in \Fk$ stands for a structured solution with a lower bound (see Definition \ref{def:structured_solution_OPT-S}). Therefore, we have $p \geq \zke{\kappa-2} T$ if $p > 0$. For every $v \leq c$, there is a tuple $(p,s,0) \in \Fke{0}$ with $p = \OPTSvs{\tI\cup\menges{\aeffc}}$ and $s \leq v$.
\end{lemma}
\begin{proof}
This lemma directly follows from the dynamic program: tuples use at most one item from every $\tIk$. For $k \in \menges{ \kappa-2, \ldots, \kappa+1}$, a tuple with $p > 0$ represents an item set that uses at least one item in $\tIe{k}, \ldots, \tIe{\kappa+1}$, and such an item has a profit of at least $\zke{\kappa-2} T$. Tuples for $k \leq \kappa-3$ with $p > 0$ are only derived from tuples that use at least one item in $\tIe{\kappa-2}, \ldots, \tIe{\kappa+1}$. If dominated tuples are not removed, the dynamic program obviously constructs tuples for all possible structured solutions with a lower bound, especially the optimum combinations for every $0 \leq v \leq c$. Removing dominated tuples does not affect the tuples that stand for the optimum item combinations so that the second property still holds.
\end{proof}
While the dynamic program above constructs the desired tuples, their number may increase dramatically until $\Fke{0}$ is obtained. We therefore use approximate dynamic programming for the tuples with profits in $[\frac{1}{4}P_0, 2 P_0]$. This method is inspired by the dynamic programming used in \cite{Kellerer2003} (see also \cite[pp.~97--112]{Kellerer2004}).

Definition \ref{def:structured_solution_OPT-S} and Lemma \ref{lemma:tuple_property} state that a tuple $(p,s,k)$ with $p > 0$ satisfies $p \geq \zke{\kappa-2} T$. Apart from $(0,0,k)$, all tuples have therefore profits in the interval $[2^{\kappa-2}T, 2P_0] \stackrel{\eqref{eq:definition_T}}{=} [\frac{1}{4}P_0,2 P_0] = [\zke{\kappa-2} T, \ldots, \zke{\kappa+1}T]$. We partition this interval into sub-intervals of length $\zke{\kappa-2} K$. We get
\begin{align*}
	[2^{\kappa-2}T, 2P_0] &= \bigcup_{\xi = 0}^{\xi_0} \intervLr{\zke{\kappa-2} T + \xi \cdot \zke{\kappa-2} K, \ \zke{\kappa-2} T + (\xi+1) \zke{\kappa-2} K} \cup \menge{2 P_0}\\
	& =:  \bigcup_{\xi = 0}^{\xi_0} \tLkx \cup \tLkxe{\xi_0 + 1}
\end{align*}
for $\xi_0 := 7 (\kappa+1) 2^{\kappa+1} - 1$. (A short calculation shows that $\zke{\kappa-2}T + (\xi_0 + 1)\zke{\kappa-2} K = 2 P_0$.) The approximate dynamic program keeps for every $\xi \in \menges{ 0, \ldots, \xi_0 + 1}$ only the tuple $(p,s,k)$ with $p \in \tLkx$ that has the smallest size $s$. The dominated tuples are removed when all tuples for $k$ have been constructed. The modified dynamic program is presented in Algorithm \ref{alg:approx_dyn_prog} and shown in Figure \ref{fig:approx_dyn_prog}. The sets of these non-dominated tuples are denoted by $\Dk$. 
For convenience, $(p(\xi), s(\xi),k) \in \Dk$ denotes the smallest tuple with a profit in $\tLkx$. We again save the backtracking information during the execution of the algorithm.
\begin{algorithm}
$\Dke{\kappa+2} := \menge{(0,0,\kappa+2)}$\;
$\backtrack{(0,0,\kappa+2)} := \emptyset$\;
\For{$k = \kappa+1, \ldots, 0$}{
	$\Dk := \emptyset$\;
	\For{$(p(\xi),s(\xi),k+1) \in \Dke{k+1}$}{
		$\Dk := \Dk \cup \menge{(p(\xi),s(\xi),k)}$\;
		$\backtrack(p(\xi),s(\xi),k) := \backtrack(p(\xi),s(\xi),k+1)$\;
	}
	\For{$\ta \in \tIk$}{
		\For(\tcp*[f]{Construction of new tuples}){$(p,s,k+1) \in \Dke{k+1} \setminus \menges{(0,0,k+1)}$}{
			$(p',s',k) := (p + p(\ta), s + s(\ta),k)$\;
			Determine $\xi'$ for $(p',s',k)$ such that $p' \in \tLkxe{\xi'}$\;
			\If{$s' < s(\xi')$ or $(p(\xi'),s(\xi'),k) = \emptyset$}{
			\tcp{Only new tuples of smaller size are kept}
				$\Dk := \Dk \setminus \menges{(p(\xi'),s(\xi'),k)}$\;
				$(p(\xi'), s(\xi'), k) := (p',s',k)$\;
				$\backtrack{(p(\xi'), s(\xi'), k)} := \klammer{(p,s,k+1), \ta}$\;
				$\Dk := \Dk \cup \menges{(p(\xi'),s(\xi'),k)}$\;
			}
		}
		\If{$k \geq \kappa-2$}{
		\tcp{Construction of (possible) tuples $(p(\ta), s(\ta), k)$ for $k \geq \kappa - 2$}
			Determine $\xi'$ for $p(\ta)$ such that $p(\ta) \in \tLkxe{\xi'}$\;
			\If{$s(\ta) < s(\xi')$ or $(p(\xi'),s(\xi'),k) = \emptyset$}{
				$\Dk := \Dk \setminus \menges{(p(\xi'),s(\xi'),k)}$\;
				$(p(\xi'), s(\xi'), k) := (p(\ta),s(\ta),k)$\;
				$\backtrack{(p(\xi'), s(\xi'), k)} := \klammer{\ta}$\;
				$\Dk := \Dk \cup \menges{(p(\xi'),s(\xi'),k)}$\;
			}
		}
	}
Remove dominated tuples from $\Dk$\;
}
\caption{The approximate dynamic programming}
\label{alg:approx_dyn_prog}
\end{algorithm}
\begin{figure}
\centering
\subfloat[The tuples in $D^{(k+1)}$\label{fig:approx_dyn_prog_a}]{\includegraphics[scale=0.6]{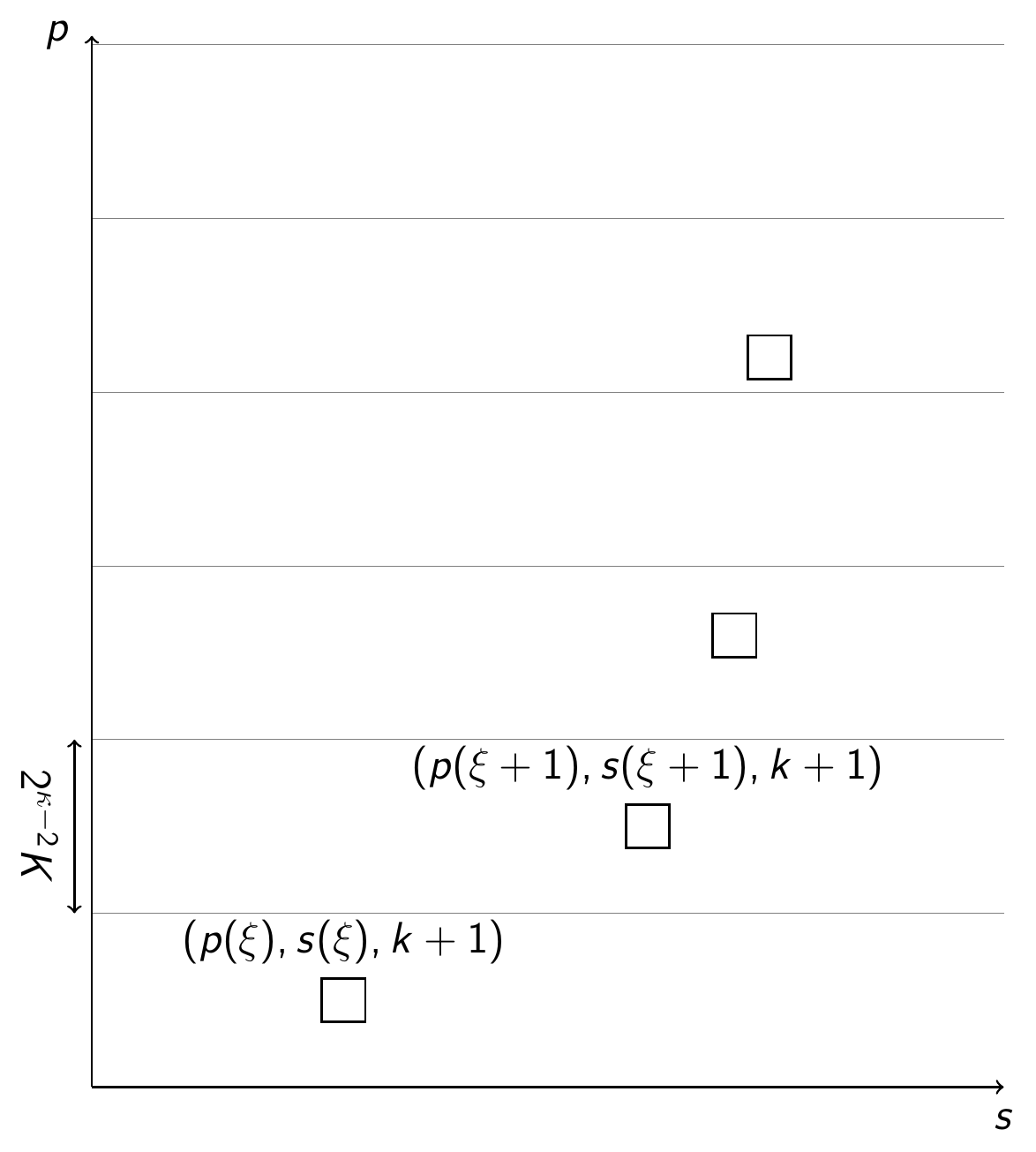}}
\hspace{3mm}
\subfloat[The new tuples are constructed with the items in $\tIk$.\label{fig:approx_dyn_prog_b}]{\includegraphics[scale=0.6]{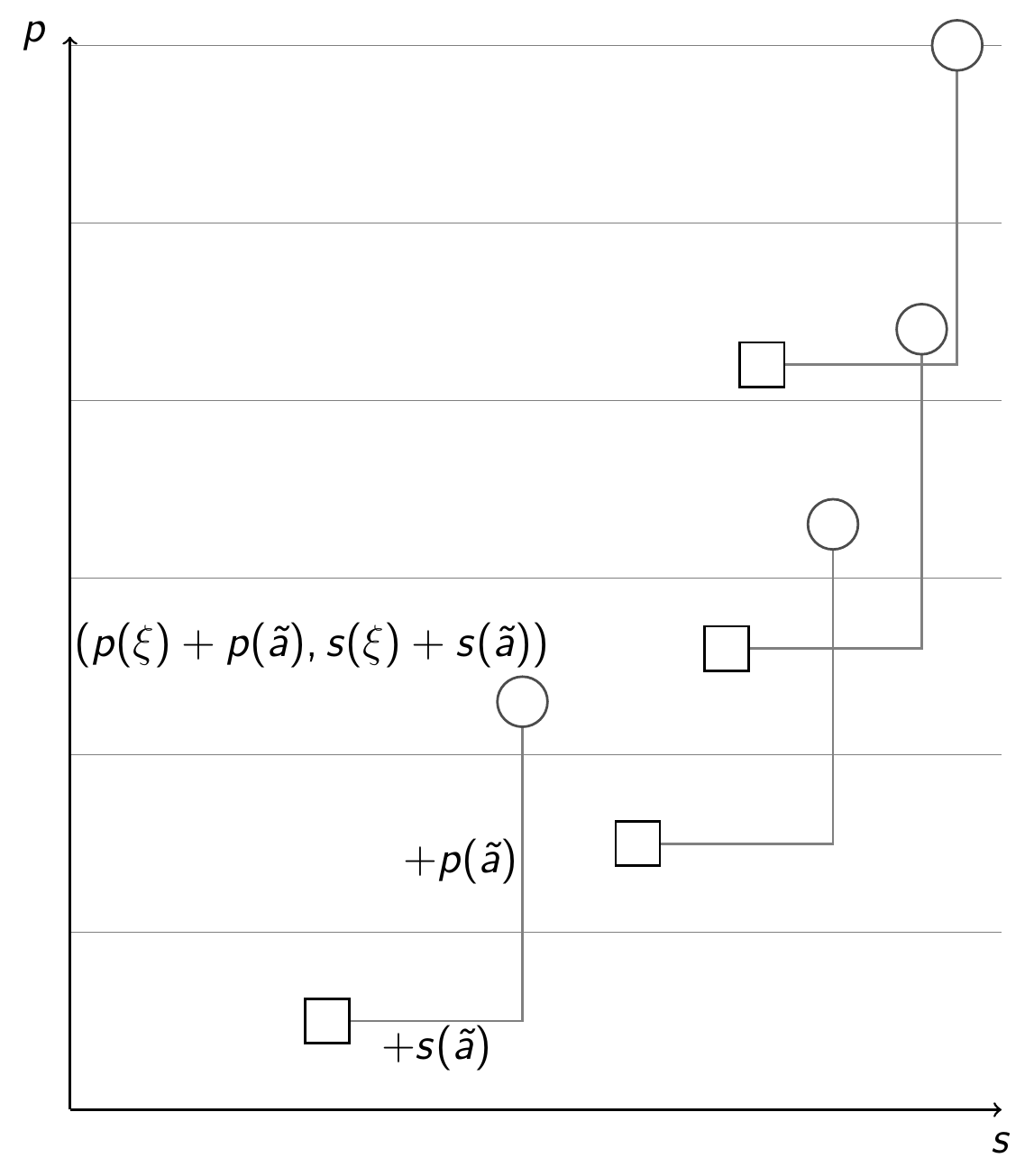}}

\subfloat[Only the tuple of smallest size is kept for every $\tLkx$, which yields $\tDk$,\ldots\label{fig:approx_dyn_prog_c}]{\includegraphics[scale=0.6]{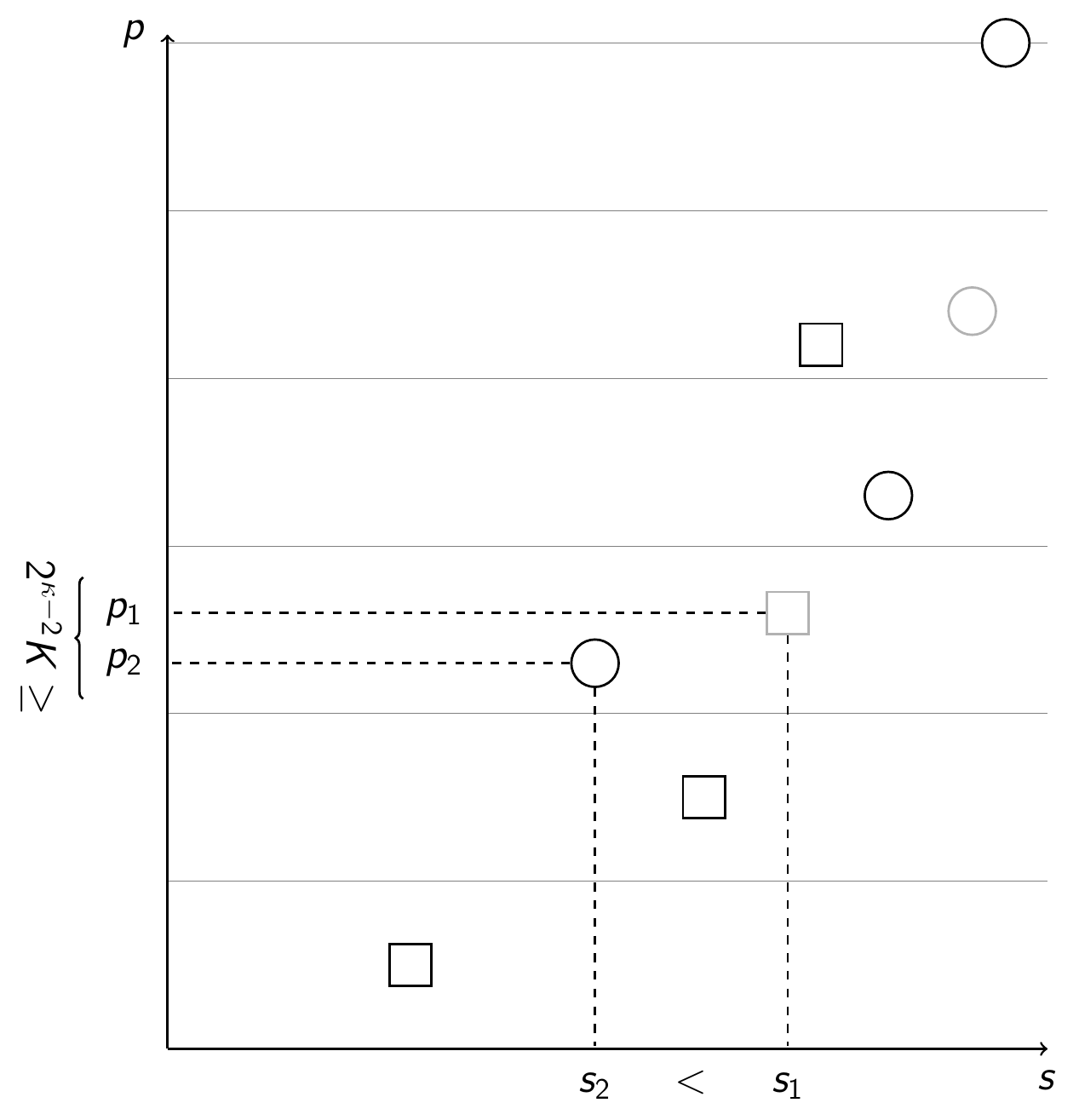}}
\hspace{3mm}
\subfloat[\ldots and removing the dominated tuples yields $\Dk$.\label{fig:approx_dyn_prog_d}]{\includegraphics[scale=0.6]{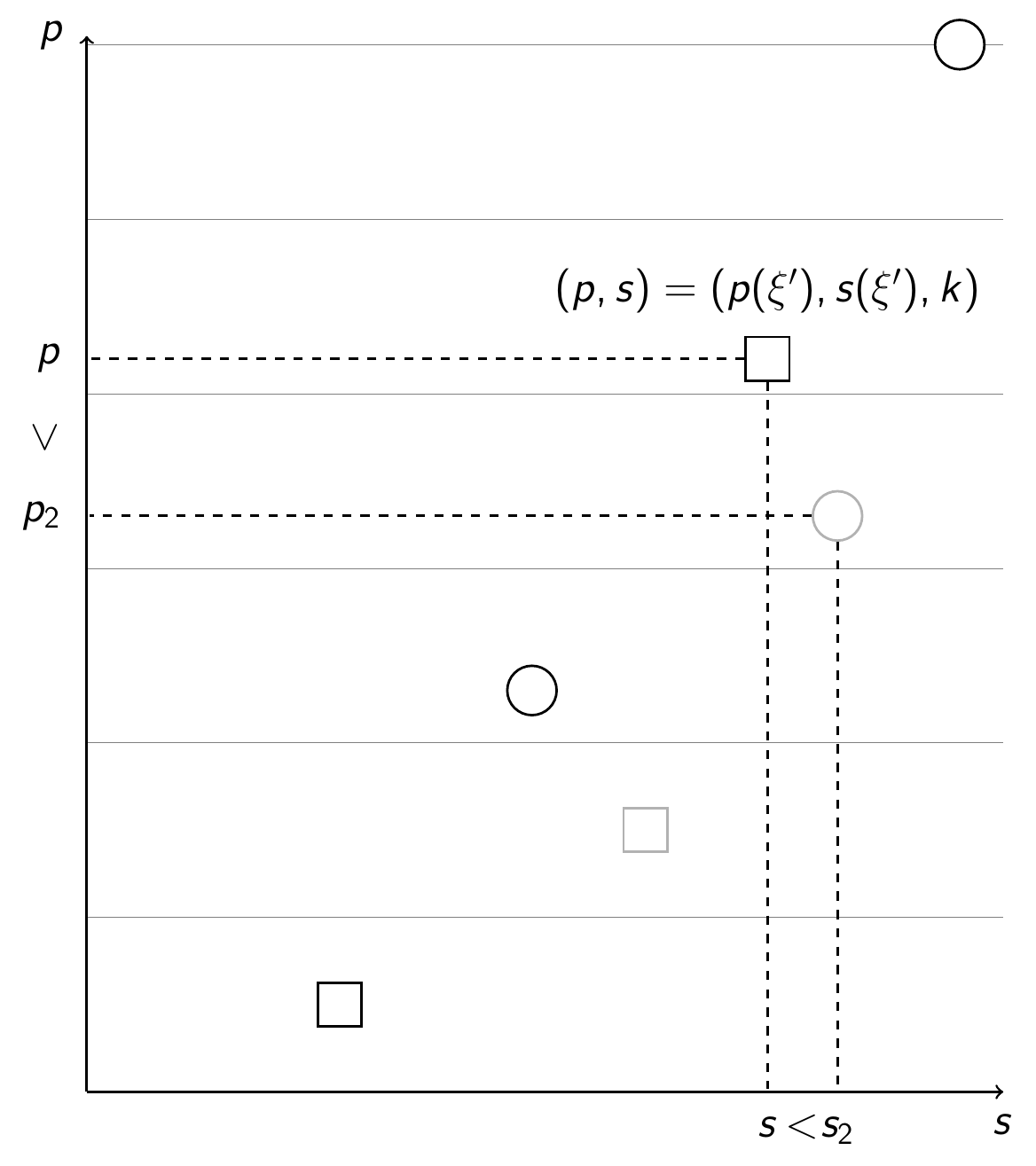}}
\caption{The principle of the approximate dynamic programming}%
\label{fig:approx_dyn_prog}%
\end{figure}

\begin{lemma} \label{lemma:tuple_property_Dk}
Let $\tDk$ be the set $\Dk$ from Algorithm \ref{alg:approx_dyn_prog} before the dominated tuples are removed. A tuple $(p,s,k) \in \tDk$ for $k = \kappa+1, \ldots, 0$ stands for a structured solution with a lower bound. Therefore, we have $p \geq \zke{\kappa-2} T$ if $p > 0$. This is also true for $(p,s,k) \in \Dk$.
\end{lemma}
\begin{proof}
The proof is almost identical to the one of Lemma \ref{lemma:tuple_property}. In fact, the proof is not influenced by keeping only the tuple of smallest size in every profit interval $\tLkx$.
\end{proof}
\begin{theorem} \label{thm:tupel-Fk_compared_to_tupel-Dk}
Let $k \in \menges{0, \ldots, \kappa+1}$. For every (non-dominated) tuple $(\bp,\bs,k) \in \Fk$, there is a tuple $(p,s,k) \in \Dk$ such that 
	\[p \geq \klammer{1 - \frac{\eps}{4}\fraclogeps}^{\kappa - k +1} \bp \quad \text{ and } \quad s \leq \bs\enspace.
\]
\end{theorem}
\begin{proof}
This statement is trivial for $(\bp,\bs,k) = (0,0,k)$ because $(0,0,k) \in \Dk$ (this tuple is never removed in the construction of $\Fk$ and $\Dk$).

Suppose now that $(\bp,\bs,k) \neq (0,0,k)$. The theorem is proved by induction for $k = \kappa + 1, \ldots, 0$.

The statement is evident for $k = \kappa + 1$. If $\aeffc$ exists (i.e.\ enough copies of $\aeff$ can be glued together without exceeding the capacity $c$), then we have 
	\[\Fke{\kappa+1} = \Dke{\kappa+1} = \menges{(0,0,\kappa+1), (p(\aeffc),s(\aeffc),\kappa+1)}\enspace.
\]
 If $\aeffc$ does not exist, then we have $\Fke{\kappa+1} = \Dke{\kappa+1} = \menges{(0,0,\kappa+1)}$.

Suppose that the statement is true for $k+1, \ldots, \kappa+1$. As defined in Lemma \ref{lemma:tuple_property_Dk}, $\tDk$ is the set $\Dk$ before the dominated tuples are removed. Let $(\bp,\bs,k) \in \Fk$.

There are two cases. In the first case, we have $(\bp,\bs,k+1) \in \Fke{k+1}$. By the induction hypothesis, there is a tuple $(p_1,s_1,k+1) \in \Dke{k+1}$ such that the inequalities $p_1 \geq \bp (1 - \frac{\eps}{4} \fraclogeps)^{\kappa-(k+1)+1} $ and $s_1 \leq \bs$ hold (see Fig.\ \ref{fig:approxdynA}\subref{fig:approxdynA_a}). Note that this implies $(p_1,s_1,k+1) \neq (0,0,k+1)$ and therefore $p_1 \geq 2^{\kappa-2} T$ by Lemma \ref{lemma:tuple_property_Dk}. Let $\xi_1$ be the index such that $p_1 \in \tLkxe{\xi_1}$. During the execution of Algorithm \ref{alg:approx_dyn_prog}, $(p_1,s_1,k+1)$ yields the tuple $(p_1,s_1,k)$, which may only be replaced in $\tDk$ by a tuple of a smaller size, but with a profit still in $\tLkxe{\xi_1}$. Thus, there must be a tuple $(p_2, s_2, k) \in \tDk$ with $s_2 \leq s_1$ and $p_2 \in \tLkxe{\xi_1}$ (see Fig.\ \ref{fig:approxdynA}\subref{fig:approxdynA_b}). 
Let now $(p,s,k) \in \Dk$ be the tuple that dominates $(p_2,s_2,k)$ (which can of course be $(p_2,s_2,k)$ itself), i.e.\ $p \geq p_2$ and $s \leq s_2$ (see Fig.\ \ref{fig:approxdynA}\subref{fig:approxdynA_c}). For the profit, we have
\begin{IEEEeqnarray*}{rCl}
p & \geq & p_2 \geq p_1 - \zke{\kappa-2} K \stackrel{p_1 \neq 0}{=} p_1 \cdot \klammer{1 - \frac{\zke{\kappa-2} K}{p_1}} \stackrel{\text{Lem.\@ }\ref{lemma:tuple_property_Dk}}{\geq} p_1 \cdot \klammer{1 - \frac{\zke{\kappa-2} K}{\zke{\kappa-2} T}}\\
& \stackrel{\eqref{eq:definition_T}, \eqref{eq:definition_K}}{=} & p_1 \cdot \klammer{1 - \frac{\eps}{4} \fraclogeps} \geq \bp \cdot \klammer{1 - \frac{\eps}{4} \fraclogeps}^{\kappa - k +1}\enspace.
\end{IEEEeqnarray*}
The lower bound on the profit is therefore true for $(p,s,k)$. We have $s \leq s_2 \leq s_1 \leq \bs$ for the bound on the size (see also Fig.\@ \ref{fig:approxdynA}\subref{fig:approxdynA_c}).
\begin{figure}%
\centering
\subfloat[Since $(\bp,\bs,k+1) \in \Fke{k+1}$, there must be a corresponding tuple $(p_1,s_1,k+1) \in \Dke{k+1}$ by the induction hypothesis whose profit can be bounded from below.\label{fig:approxdynA_a}]{\includegraphics[scale=0.6]{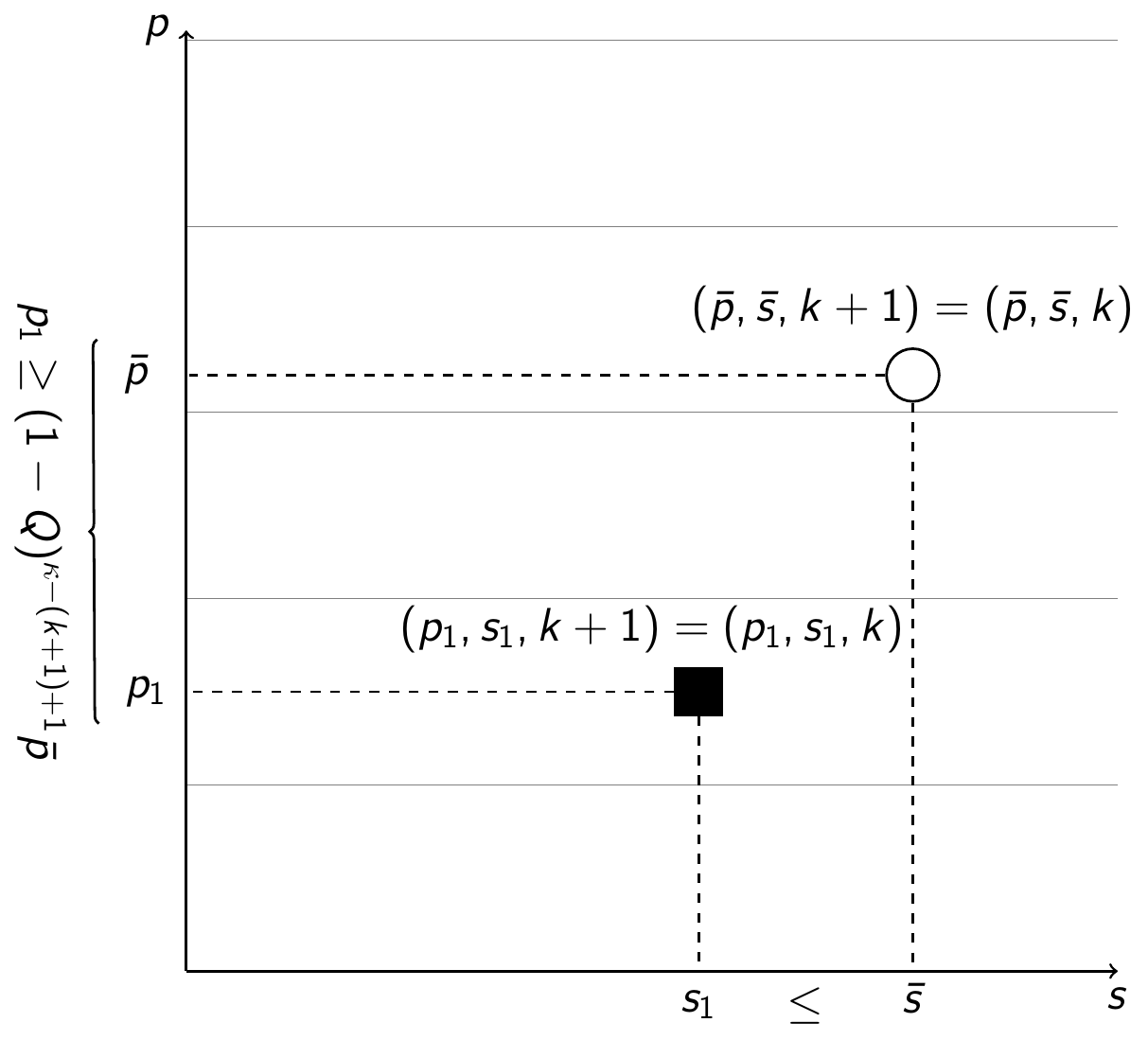}}%
\hspace{1mm}
\subfloat[By construction, there must be a tuple $(p_2, s_2,k) \in \tDk$ with a profit in the same interval $\tLkxe{\xi_1}$ as $(p_1,s_1,k)$. This makes it possible to bound $p_2$ from below.\label{fig:approxdynA_b}]{\includegraphics[scale=0.6]{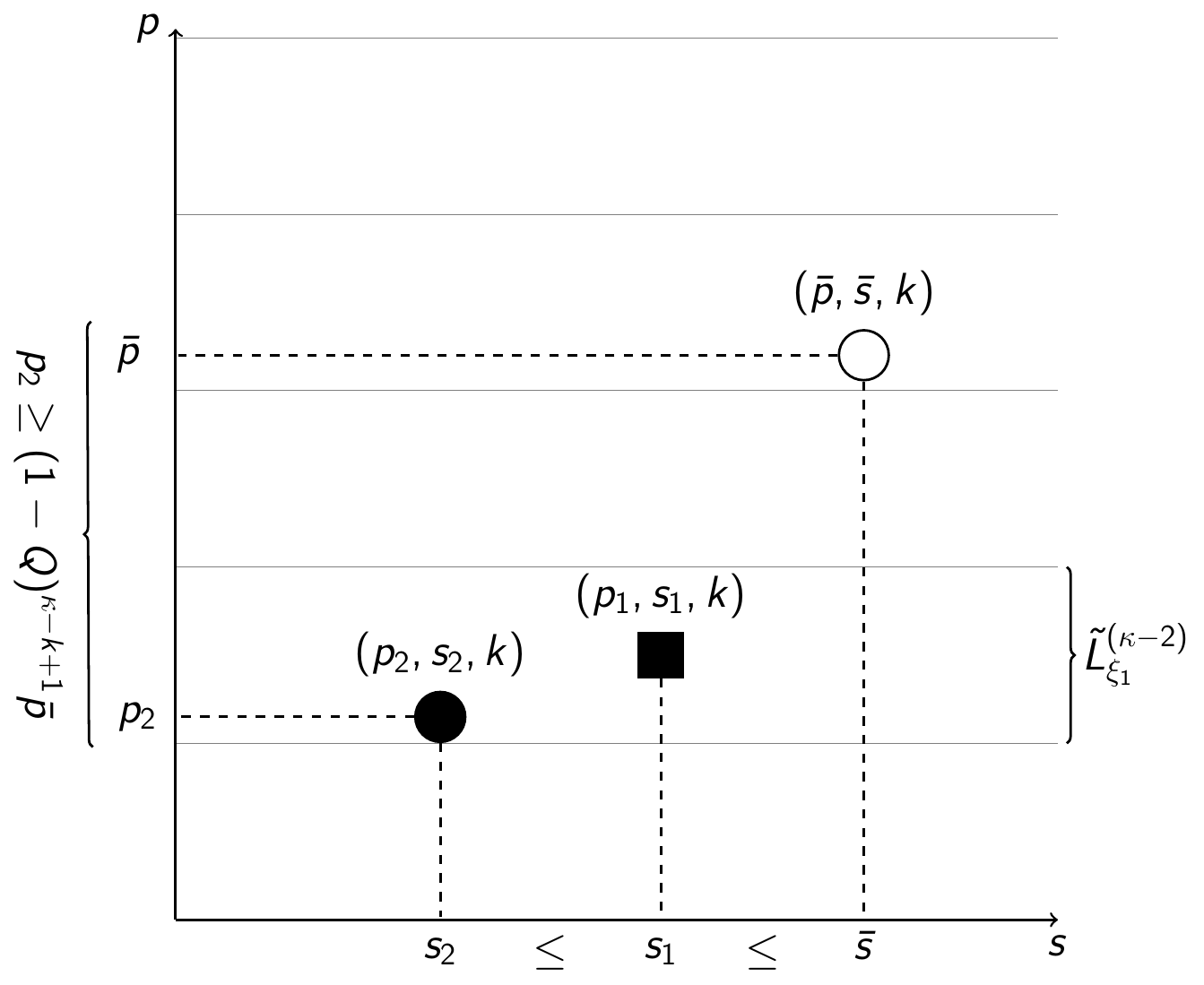}}%

\subfloat[There may be a tuple $(p,s,k) \in \Dk$ that dominates $(p_2,s_2,k)$. Since $p \geq p_2$, the bound still holds.\label{fig:approxdynA_c}]{\includegraphics[scale=0.6]{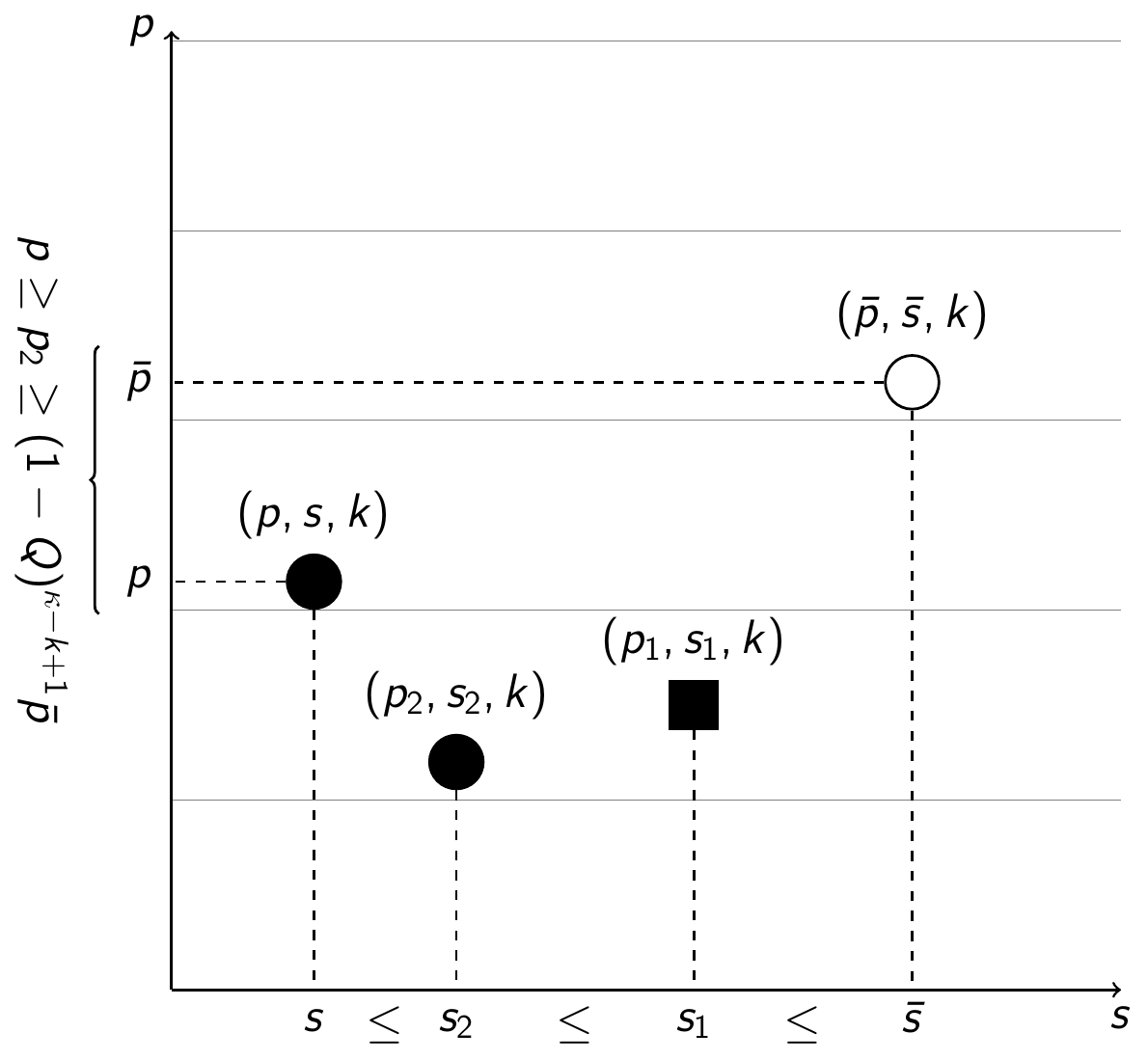}}%
\caption{The first case of the proof for Theorem \ref{thm:tupel-Fk_compared_to_tupel-Dk}: we have $(\bp,\bs,k) \in \Fk$ and also $(\bp,\bs,k+1)\in \Fke{k+1}$. We set $Q := (1-\frac{\eps}{4} \fraclogeps)$.}%
\label{fig:approxdynA}%
\end{figure}

Consider now the second case where $(\bp,\bs,k) \in \Fk$, but $(\bp,\bs,k+1) \notin \Fke{k+1}$. Therefore, $(\bp,\bs,k)$ is a new (non-dominated) tuple with $(\bp,\bs,k) = (\tp + p(\ta), \ts + s(\ta),k)$ for the right item $\ta \in \tIk$ and tuple $(\tp,\ts,k+1) \in \Fke{k+1}$. By the induction hypothesis, there must be a tuple $(p_1,s_1,k+1) \in \Dke{k+1}$ such that $p_1 \geq \tp (1 - \frac{\eps}{4} \fraclogeps)^{\kappa-(k+1)+1} $ and $s_1 \leq \ts$ (see Fig.\@ \ref{fig:approxdynB}\subref{fig:approxdynB_a}).
Thus, the following inequality holds:
\begin{IEEEeqnarray*}{rCl}
p_1 + p(\ta) &\geq& p(\ta) + \tp \cdot \klammer{1 - \frac{\eps}{4}\fraclogeps}^{\kappa-(k+1)+1} \geq  \klammer{p(\ta) + \tp} \cdot \klammer{1 - \frac{\eps}{4}\fraclogeps}^{\kappa-(k+1)+1}\\
& = & \bp \cdot \klammer{1 - \frac{\eps}{4}\fraclogeps}^{\kappa-(k+1)+1}\enspace.
\end{IEEEeqnarray*}
There are two possibilities: either $k \geq \kappa-2$, i.e.\ $p(\ta) \geq \zke{\kappa-2} T$ holds, and $p_1 + p(\ta) \geq \zke{\kappa-2} T$ directly follows. Otherwise, we have $k \leq \kappa-3$. Then, the identity $(\bp,\bs,k) = (\tp + p(\ta), \ts + s(\ta),k) \neq (0,0,k)$ implies that $(\tp,\ts,k+1) \neq (0,0,k+1)$ holds because the tuple $(0,0,k+1)$ is not used to form any new tuple in $\tDk$ and therefore in $\Dk$. This again implies that $p_1 \neq 0$ and therefore $p_1 + p(\ta) \geq p_1 \geq \zke{\kappa-2} T$ as seen in Lemma \ref{lemma:tuple_property_Dk}.

Thus, there is an index $\xi_1$ such that $p_1 + p(\ta) \in \tLkxe{\xi_1}$. Similar to above, the tuple $(p_1 + p(\ta), s_1 + s(\ta) , k)$ is formed during the construction of $\tDk$ (see Fig.\@ \ref{fig:approxdynB}\subref{fig:approxdynB_b}). It may only be replaced by a tuple of smaller size. Hence, 
there must be $(p_2,s_2,k) \in \tDk$ with $p_2 \in \tLkxe{\xi_1}$. Let $(p,s,k) \in \Dk$ be the tuple that dominates $(p_2,s_2,k)$ (see Fig.\@ \ref{fig:approxdynB}\subref{fig:approxdynB_c}). We get
\begin{IEEEeqnarray*}{rCl}
p & \geq & p_2 \geq p_1 + p(\ta) - \zke{\kappa-2} K \stackrel{p_1 + p(\ta) \neq 0}{=} ( p_1 + p(\ta) ) \cdot \klammer{1 - \frac{\zke{\kappa-2} K}{p_1 + p(\ta)}} \\
&\stackrel{p_1 + p(\ta) \geq \zke{\kappa-2} T}{\geq}& (p_1 + p(\ta)) \cdot \klammer{1 - \frac{\zke{\kappa-2} K}{\zke{\kappa-2} T}}
 \stackrel{\eqref{eq:definition_T}, \eqref{eq:definition_K}}{=} (p_1 + p(\ta)) \cdot \klammer{1 - \frac{\eps}{4} \fraclogeps}\\ 
&\geq &\bp \cdot \klammer{1 - \frac{\eps}{4} \fraclogeps}^{\kappa - k +1}\enspace.
\end{IEEEeqnarray*}
We have similar to above $s \leq s_2 \leq s_1 + s(\ta) \leq \ts + s(\ta) = \bs$ for the bound on the size (see also Fig.\@ \ref{fig:approxdynB}\subref{fig:approxdynB_c}).
\begin{figure}%
\centering
\subfloat[Since $(\bp,\bs,k+1) \notin \Fke{k+1}$, there must be an item $\ta$ such that $(\bp,\bs,k) = (\tp + p(\ta),\ts + s(\ta), k)$ for a tuple $(\tp,\ts,k+1) \in \Fke{k+1}$. By the induction hypothesis, there must be a tuple $(p_1,s_1,k+1) \in \Dke{k+1}$ whose profit can be bounded from below.\label{fig:approxdynB_a}]{\includegraphics[scale=0.6]{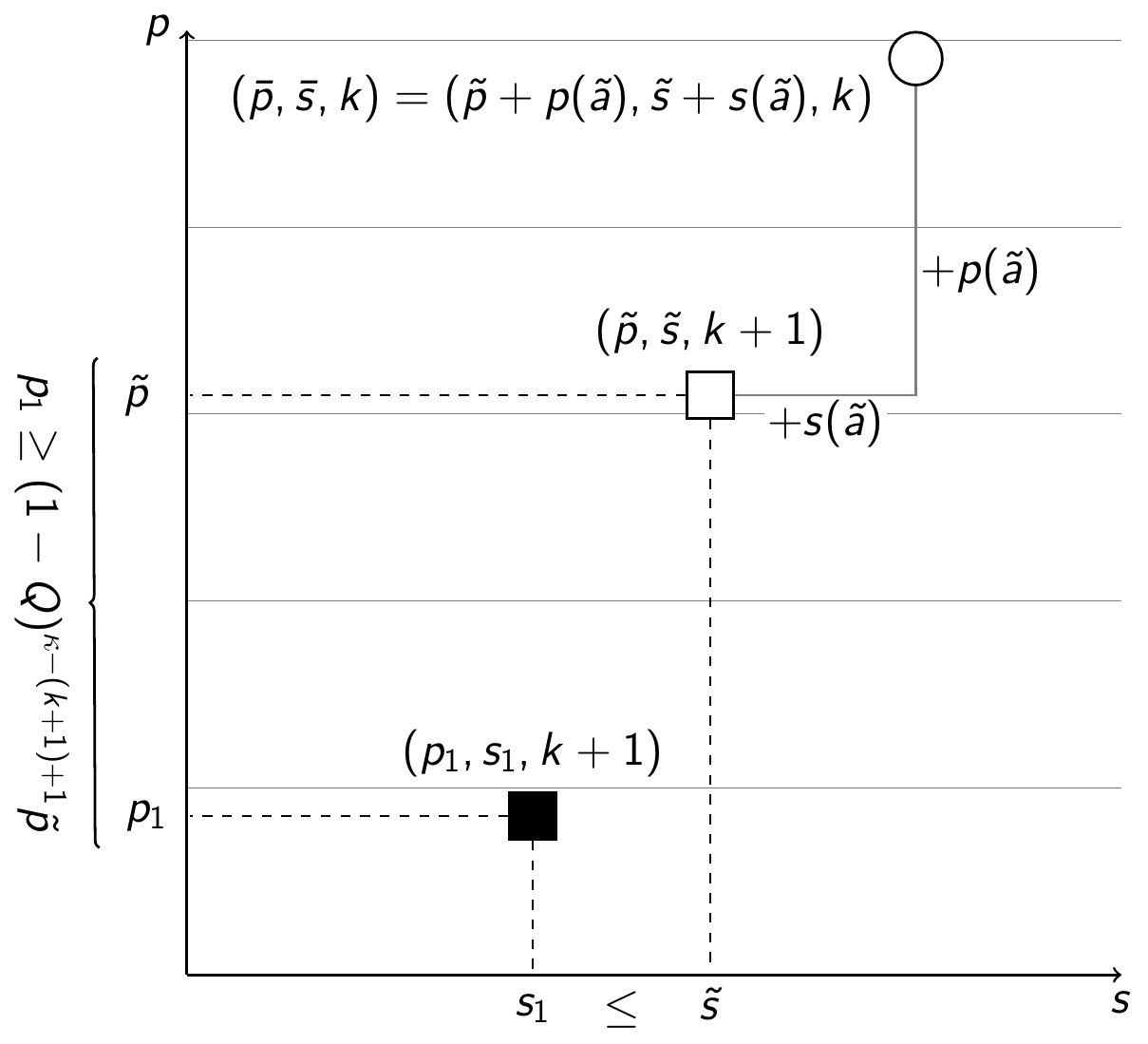}}%
\hspace{1mm}
\subfloat[The tuple $(p_1+p(\ta),s_1+s(\ta),k)$ is constructed during the execution of the dynamic program.\label{fig:approxdynB_b}]{\includegraphics[scale=0.6]{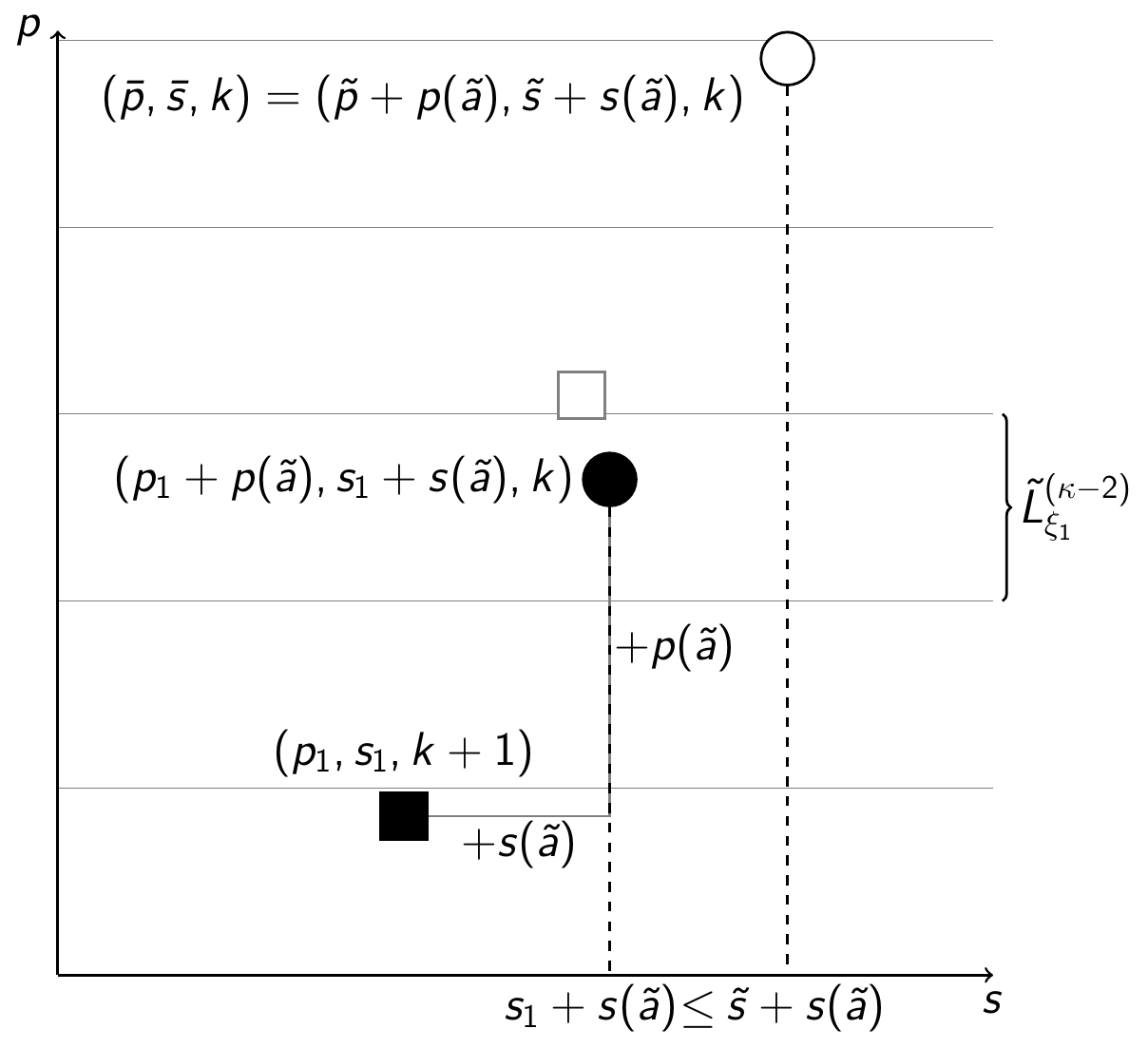}}%

\subfloat[As in the first case, there must be a tuple $(p,s,k) \in \Dk$ whose profit can be bounded as desired. Here, $(p_2,s_2,k)$ is not dominated, i.e.\ $(p,s,k) = (p_2,s_2,k)$. \label{fig:approxdynB_c}]{\includegraphics[scale=0.6]{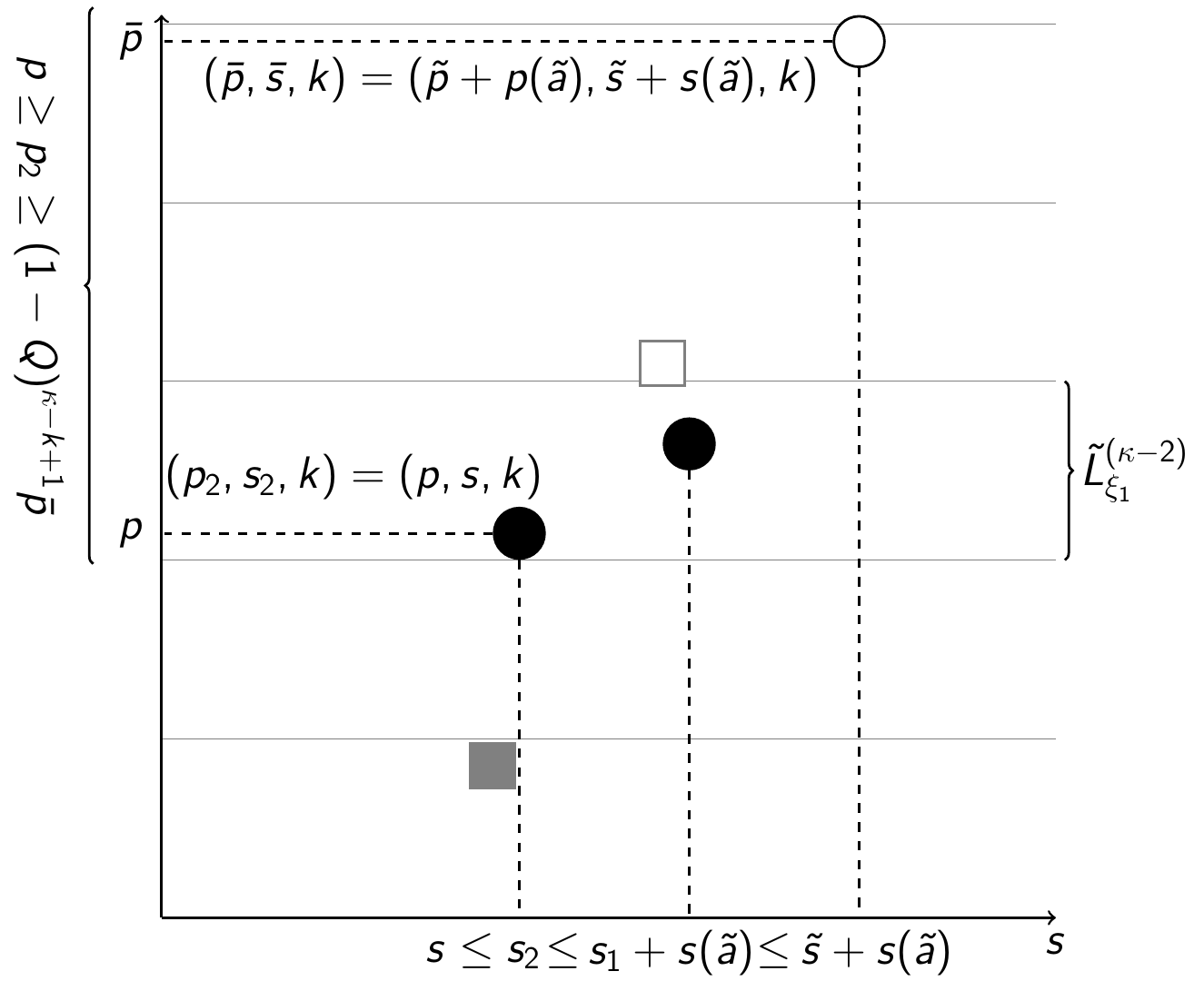}}%
\caption{The second case of the proof for Theorem \protect\ref{thm:tupel-Fk_compared_to_tupel-Dk}: we have $(\bp,\bs,k) \in \Fk$, but $(\bp,\bs,k+1)\notin \Fke{k+1}$. We set $Q := (1-\frac{\eps}{4} \fraclogeps)$.}%
\label{fig:approxdynB}%
\end{figure}
\end{proof}
\begin{remark} \label{remark:dynamic_programming_error_because_of_structure}
As can be seen, the proof of Theorem \ref{thm:tupel-Fk_compared_to_tupel-Dk} is only possible because it is guaranteed that $p_1$ or $p_1 + p(\ta)$ is at least $\zke{\kappa-2} T$. In fact, this is achieved by the construction of the glued item set $\tI$ with its structured solution (Theorem \ref{thm:solution_quality_tIk}). Hence, we can prove Lemma \ref{lemma:two_cases}, and with the introduction of $\aeffc$, we have the structure property of Definition \ref{def:structured_solution_OPT-S} with a corresponding solution (Theorem \ref{thm:structured_solution_exists}). This shows that $p_1 \geq 2^{\kappa-2}T$ or $p_1 + p(\ta) \geq 2^{\kappa-2} T$ (see also Lemma \ref{lemma:tuple_property} and \ref{lemma:tuple_property_Dk}). Without the structure, a dynamic program like Algorithm \ref{alg:approx_dyn_prog} would also have to generate tuples $(p,s,k)$ with $p < \zke{\kappa-2} T$ for $k \leq \kappa-3$. 
Hence, we would need for the same approximation ratio profit sub-intervals like $\tLkx$ with a smaller length than $2^{\kappa-2} K$, and we would have to save more tuples. Both would increase the asymptotic running time and space complexity as can be seen in the proof of Theorem \ref{thm:construction_Dk_running_time}.
\end{remark}
\begin{corollary} \label{corollary:existence_tuple_Dke0}
For every $v \leq c$, there is a tuple $(p,s,0) \in \Dke{0}$ such that $s \leq v$ and
	\[p \geq \klammer{1 - \frac{\eps}{4}\fraclogeps}^{\kappa + 1} \OPTSv{\tI \cup \menge{\aeffc}}\enspace.
\]
\end{corollary}
\begin{proof}
Lemma \ref{lemma:tuple_property} states that there is a tuple $(\bp,\bs,0) \in \Fke{0}$ with $\bp = \OPTSvs{\tI \cup \menge{\aeffc}}$ and $\bs \leq v$. Theorem \ref{thm:tupel-Fk_compared_to_tupel-Dk} implies that there is a tuple $(p,s,0) \in \Dke{0}$ with the desired property.
\end{proof}
\begin{theorem}\label{thm:construction_Dk_running_time}
Algorithm \ref{alg:approx_dyn_prog} constructs all tuple sets $\Dk$ for $k = \kappa+1, \ldots, 0$ in time $\Ohs{\frepse{2} \log^3 \freps}$. The space needed for the algorithm and to save the $\Dk$ as well as the backtracking information is in $\Ohs{\freps \log^2\freps}$.
\end{theorem}
\begin{proof}
Let us first bound the space complexity. The profit interval $[\frac{1}{4} P_0, 2 P_0]$ is partitioned into $\Ohs{\xi_0}$ intervals $\tLkx$. The set $\Dk$ saves at most one tuple with the corresponding backtracking information for every $\tLkx$ or the information that a tuple does not exist. Thus, the space needed for all $\Dk$ and the corresponding backtracking data is in $\Ohs{\kappa \cdot \xi_0} = \Ohs{\kappa \cdot \klammers{\kappa \zke{\kappa}}} =\Ohs{\log(\freps) \cdot (\log(\freps) \freps )} = \Ohs{\freps \log^2\freps}$. All other information of the algorithm is only temporarily saved and needs $\Oh{1}$.

The loops dominate the running time. Apart from removing the dominated tuples, they need in total 
\begin{IEEEeqnarray*}{rCl}
	\Oh{\kappa \cdot \klammer{\xi_0 + \xi_0 \cdot \betrs{\tIk} + \betrs{\tIk}}} &\stackrel{\text{Thm.~\ref{thm:construction_tI_aeffc_running-time}}}{=}& \Oh{\log\klammer{\freps} \klammer{ \freps \log\klammer{\freps} \cdot \freps \log\klammer{\freps} } } \\
	&=& \Oh{\frepse{2} \log^3 \freps} \enspace.
\end{IEEEeqnarray*}
As stated in \cite{Lawler1979} and \cite[Lemma 5]{Jansen2013}, non-dominated tuples $(p,s,k)$ can be removed in linear time in the number of tuples if the tuples are different and sorted by profit. This is the case because every tuple in $\Dk$ is stored in an array sorted according to the corresponding $\xi$. The total time  to remove the dominated tuples from all $\Dk$ is therefore in $\Ohs{\kappa \cdot \xi_0} = \Ohs{\freps \log^2\freps}$, which is dominated by the overall running time.
\end{proof}

%% file: complete_algorithm.tex
\section{The Algorithm}
We can now put together the entire approximation algorithm.
\begin{algorithm}
\KwIn{Item set $I$}
\KwOut{Profit $P$, solution set $J$}
Determine $P_0$ and define $T, K$\;
Partition the items into $I_L$ and $I_S$ and find $\aeff$\;
\If{Item $a$ with $p(a) = 2 P_0$ found during the partitioning}
{\Return{$2 P_0$, $\menges{a}$}\;}
Reduce $I_L$ to $\ILred$ with Algorithm \ref{alg:find_a-k-g}\;
Construct $\tI$ with Algorithm \ref{alg:construction_tIk} and the item $\aeffc$\;
\If{$p(\taege{\kappa}{0}) = P_0$ and $s(\taege{\kappa}{0}) \leq \frac{c}{2}$}
{ Recursively undo the gluing of $\taege{\kappa}{0}$ to get the item set $J'$. Let $J$ be the set consisting of two copies of every item in $J'$\;
  \Return{$2 P_0$, $J$} \;}
Construct with Algorithm \ref{alg:approx_dyn_prog} the tuple sets $\Dke{\kappa+1}, \ldots, \Dke{0}$\;
Find $(p,s,0) \in \Dke{0}$ such that $P := p + \OPTvv{\menges{\aeff}}{c-s} = \max_{(p',s',0) \in \Dke{0}} p' + \OPTvv{\menges{\aeff}}{c-s'}$\;
Backtrack the tuple $(p,s,0)$ to find the corresponding structured solution with a lower bound $J'\subset \tI \cup \menges{\aeffc}$\;
Recursively undo the gluing of all $\ta \in J'$ and add these items to the solution set $J$\;
Add the items of $\OPTvvs{\menges{\aeff}}{c-s}$ to $J$\;
\Return{$P$, $J$ \;}
\caption{The complete algorithm}\label{alg:main}
\end{algorithm}
\begin{theorem}
Algorithm \ref{alg:main} finds a solution of value at least $(1-\eps) \OPT(I)$.
\end{theorem}
\begin{proof}
The algorithm returns a feasible solution: $(p,s,0)$ represents an item set of size $s$. If items $\ta \in \tI$ derived from gluing are part of the solution, their ungluing does not change the total size nor the total profit (see Remark \ref{remark:items_takg_represent_original_items}).

We prove the solution quality. First, the algorithm considers the two special cases listed at the beginning of Section \ref{sec:dynamic_programming}. Each of them returns a solution of profit $2 P_0$ so that $\OPT(I) = 2 P_0$ (see Theorem \ref{thm:1-2_approximation_P_0} and Lemma \ref{lemma:estimates_I-L_to_tI}). If the special cases do not yield a solution, we are in the third case. Let $v$ be the volume from Theorem \ref{thm:structured_solution_exists}. 
Corollary \ref{corollary:existence_tuple_Dke0} guarantees the existence of one $(p,s,0) \in \Dke{0}$ with $s \leq v$ such that
	\[p \geq \klammer{1 - \frac{\eps}{4}\fraclogeps}^{\kappa + 1} \OPTSv{\tI \cup \menge{\aeffc}}\enspace.
\]
Moreover, we have $\OPTvvs{\menges{\aeff}}{c-s} \geq \OPTvvs{\menges{\aeff}}{c-v}$ because $c-s \geq c-v$. 
Thus, the following inequality holds for this $(p,s,0)$:
\begin{IEEEeqnarray*}{rCl}
\IEEEeqnarraymulticol{3}{l}{p + \OPTvv{\menges{\aeff}}{c-s}}\\ \quad
 &\geq& \klammer{1 - \frac{\eps}{4}\fraclogeps}^{\kappa + 1} \OPTSv{\tI \cup \menge{\aeffc}} +  \OPTvv{\menges{\aeff}}{c-v}\\
& \geq &  \klammer{1 - \frac{\eps}{4}\fraclogeps}^{\kappa + 1} \klammer{\OPTSv{\tI \cup \menge{\aeffc}} + \OPTvv{\menges{\aeff}}{c-v}}\\
&\stackrel{\text{Thm.~\ref{thm:structured_solution_exists}}}{\geq}& \klammer{1 - \frac{\eps}{4}\fraclogeps}^{2\kappa + 2} \OPT(I) - \klammer{1 - \frac{\eps}{4}\fraclogeps}^{\kappa + 1}T\\
&\geq& \klammer{1 - \frac{\eps}{4}\fraclogeps}^{2\kappa + 2} \OPT(I) - T\\
&\stackrel{\eqref{eq:definition_T}}{\geq}& \klammer{1 - \frac{\eps}{4} \efraclogeps{2\kappa+2}} \OPT(I) - \frac{1}{2} \eps P_0\\
& \geq & \klammer{1 - \frac{\eps}{4} \cdot \efraclogeps{2 \cdot \klammer{\log_2(\frac{2}{\eps}) + 1}}} \OPT(I) - \frac{1}{2} \eps \OPT(I)\\
&=& \klammer{1 - \eps} \OPT(I)\enspace.
\end{IEEEeqnarray*}
Taking the maximum over all $(p,s,0) \in \Dke{0}$ therefore yields the desired solution. Note that we have used $(1-\delta)^k \geq (1- k \cdot \delta)$ for $\delta < 1$.
\end{proof}
\begin{remark}
The total bound on the approximation ratio is mainly due to the exponent $2 \kappa + 2$, i.e.\ that we make the multiplicative error of $(1-\frac{\eps}{4}\fraclogeps)$ only $2 \kappa + 2$ times. Such an error occurs when $I$ is replaced by $\ILred$ at the beginning (Lemma \ref{lemma:solution_quality_I-L_a-eff}), in each of the $\kappa$ iterations in which $\tI$ is constructed (Theorem \ref{thm:solution_quality_tIk}), and in $\kappa+1$ of the $\kappa+2$ iterations of the dynamic program (Theorem \ref{thm:tupel-Fk_compared_to_tupel-Dk} and Corollary \ref{corollary:existence_tuple_Dke0}). The error of the dynamic program can be bounded because the structured solution with a lower bound has at least one item of profit at least $2^{\kappa-2}T$ (see the second property of Definition \ref{def:structured_solution_OPT-S} and Remark \ref{remark:dynamic_programming_error_because_of_structure}). 
\end{remark}
\begin{theorem}
The algorithm has a running time in $\Ohs{n + \frepse{2} \log^3 \freps}$ and needs space in $\Ohs{n + \freps \log^2 \freps}$.
\end{theorem}
\begin{proof}
Determining $P_0$, constructing $I_L$ and $I_S$ as well as finding $\aeff$ can all be done in time and space $\Oh{n}$ as stated in Theorems \ref{thm:1-2_approximation_P_0} and \ref{thm:construction_I-L_and_a-eff}. The definition of $T$ and $K$ in time and space $\Oh{1}$ is obvious. It is also clear that an item $p(a) = 2 P_0$ can directly be found during the construction of $I_L$ such that the first if-condition does not influence the asymptotic running time.

Algorithm \ref{alg:find_a-k-g} returns the set $\ILred$ in time $\Ohs{n + \freps \log^2 \freps}$ and space $\Ohs{\freps \log^2 \freps}$ (see Theorem \ref{thm:number_items_I-L-red_time_space_constructing_I-L-red}).

Algorithm \ref{alg:construction_tIk} constructs the $\tIk$ and $\tI$ in time $\Ohs{\frepse{2} \log^3 \freps}$ and space $\Ohs{\freps \log^2\freps}$ as explained in Theorem \ref{thm:construction_tI_aeffc_running-time}, which clearly dominates the construction of $\aeffc$ in $\Oh{1}$.

The second if-condition can be checked in $\Oh{1}$. The running time for undoing the gluing will be determined at the end of the proof.

Algorithm \ref{alg:approx_dyn_prog} constructs the sets $\Dk$ in time $\Ohs{\frepse{2} \log^3 \freps}$ and space $\Ohs{\freps \log^2 \freps}$ (see Theorem \ref{thm:construction_Dk_running_time}). For one tuple $(p',s',0)$, the corresponding $\OPTvv{\menges{\aeff}}{c-s'}$ can be found in $\Oh{1}$ by computing $\unterks{\frac{c-s'}{s(\aeff)}} \cdot p(\aeff)$. Thus, finding the best tuple $(p,s,0)$ can be done in $\Ohs{|\Dke{0}|} = \Ohs{\xi_0} = \Ohs{\freps \log \freps}$. Since only the currently best tuple $(p,s,0)$ has to be saved, the space needed is in $\Oh{1}$.

The backtracking for the tuple $(p,s,0)$ needs time in $\Ohs{\kappa} = \Ohs{\log \freps}$: the backtracking information $\backtrack{(p',s',k)}$ for $k = 0,\ldots,\kappa+1$ states whether the tuple was formed by adding an item $\ta \in \tIk$ and with which tuple $(p'',s'',k+1)$ to continue. Hence, the item set $J'$ also has at most $\Ohs{\log \freps}$ items in $\tI \cup \menges{\aeffc}$, which bounds the storage space needed.

To conclude, the time and space for the ungluing still have to be bounded. Consider one item $\ta \in \tI$. The backtracking information $\backtrack{(\ta)}$ returns two items $(\ba_1,\ba_2)$ (with $\ba_1, \ba_2 \in \ILred \cup \tI$) on which the backtracking can be recursively applied. The recursive ungluing of the items can be represented as a binary tree where the root is the original item $\ta$ and the (two) children of each node are the items $(\ba',\ba'')$ returned by the backtracking information. The leaves of the tree are the original items in $\ILred$. This binary tree obviously has a height in $\Ohs{\kappa}$ because the children $(\ba',\ba'')$ for one $\ba \in \tIk$ are in $\tIe{k-1} \cup \ILred$. A binary tree of height $\Ohs{\kappa} = \Ohs{\log \freps}$ has at most $\Ohs{\freps}$ nodes. The backtracking or ungluing of $\ta$ can therefore be done in time and space $\Ohs{\freps}$, which also includes saving the items $\ba \in \ILred$ of which $\ta$ is composed. Since $J'$ has $\Ohs{\log \freps}$ items, the original items $\ILred$ of the approximate solution can be found in time and space $\Ohs{\freps \log \freps}$. This also dominates the time to undo the gluing of $\taege{\kappa}{0}$ should the body of the second if-condition be executed.

Similar to above, the number of items $\aeff$ for $\OPTvv{\menges{\aeff}}{c-s}$ can be 
found in $\Oh{1}$. To sum up, Algorithm \ref{alg:main} has the stated running time 
and space complexity.
\end{proof}
\section{Concluding Remarks}
The most important steps in this algorithm are the creation of the item set $
\tI$ by gluing and the introduction of $\aeffc$. This guarantees the 
existence of an approximate structured solution with a lower bound (see Definition \ref{def:structured_solution_OPT-S}). Therefore, the approximate dynamic 
program has to store less tuples $(p,s,k)$ than in the case without the structure.

We \cite{Kraft2015} have extended our algorithm to the Unbounded Knapsack Profit with 
Inversely Proportional Profits (UKPIP) introduced in \cite{Jansen2013}. Here, 
several knapsack sizes $0 < c_1 < \ldots < c_M = 1$ are given, and the profit 
of an item counts as $\nicefrac{p_j}{c_l}$ if packed in $c_l$. The goal is to 
find the best knapsack size and corresponding solution of maximum profit. 
UKPIP is used for column generation in our AFPTAS for Variable-Sized Bin 
Packing \cite{Jansen2012} where several bin sizes are given and the goal is 
to minimize the total volume of the bins used. The faster FPTAS for UKPIP 
yields a faster AFPTAS for Variable-Sized Bin Packing \cite{Kraft2015}. 

There are interesting open questions. As stated in Subsection \ref{subsec:our_result}, the space complexity is a more serious bottleneck than the running time. Recently, Lokshtanov and Nederlof \cite{Lokshtanov2010} showed that the 0-1 Knapsack Problem and the Subset Sum Problem have a pseudo-polynomial time and only polynomial space algorithm. Subset Sum is a special case of the Knapsack Problem where the profit of an item is equal to its size, i.e.\ $p_j = s_j$. Moreover, it was shown that Unary Subset Sum is in Logspace \cite{Kane2010,Elberfeld2010}. G{\'a}l et al.\@ \cite{Gal2014} described an FPTAS for Subset Sum whose space complexity is in $\Ohs{\freps}$, i.e.\ which does not depend on the actual input size, and whose running time is in $\Ohs{\freps n (n + \log n + \log \freps)}$. Can any of these results be further extended to improve the space complexity of an UKP FPTAS? 

Finally, it is open whether the ideas presented in this paper can be extended to the normal 0-1 KP or other KP variants as well as used for column generation of other optimization problems. The currently 
fastest known algorithm for 0-1 KP is due to Kellerer and Pferschy \cite{Kellerer1999,Kellerer2004a,Kellerer2004}. We mention in closing that by using the same approach similar improved approximation algorithms can be expected for various Packing and Scheduling Problems, e.g.\ for Bin Covering, Bin Packing with Cardinality Constraints, Scheduling Multiprocessor Tasks and Resource-constrained Scheduling.